\documentclass[12pt]{article}
\usepackage{fullpage}
\usepackage{mathtools}
\usepackage[T1]{fontenc}
\usepackage{bm}
\usepackage{ae,aecompl}
\usepackage[dvips]{graphicx}
\usepackage{amssymb}
\usepackage{amsmath}
\usepackage{mathrsfs}
\usepackage[round,authoryear]{natbib}
\usepackage{color}
\usepackage{array}
\usepackage[pdfborder=false]{hyperref}
\usepackage[titletoc,title]{appendix}

\usepackage{sectsty}

\newcommand{\clw}{{\cal W}}

\definecolor{webgreen}{rgb}{0,0.4,0}
\definecolor{webbrown}{rgb}{0.6,0,0}
\definecolor{purple}{rgb}{0.5,0,0.25}
\definecolor{darkblue}{rgb}{0,0,0.7}
\definecolor{darkred}{rgb}{0.7,0,0}
\definecolor{darkgreen}{rgb}{0,0.7,0}
\hypersetup{colorlinks,citecolor=darkred,filecolor=black,linkcolor=darkblue,urlcolor=webgreen,pdfpagemode=none,
pdfstartview=FitH}
\newcommand{\ignore}[1]{}

\newtheorem{prop}{{\sc Proposition}}

\newtheorem{theorem}{{\sc Theorem}}
\newtheorem{defn}{{\sc Definition}}
\newtheorem{obs}{{\sc Observation}}

\newtheorem{example}{{\sc Example}}

\newtheorem{assumption}{{\sc Assumption}}
\newenvironment{proof}{\noindent {\bf \sl Proof\/}:\enspace}
{\hfill $\blacksquare{}$ \vspace{12pt}}
\sectionfont{\centering\normalfont\scshape}
\subsectionfont{\centering\normalfont\scshape}

\begin{document}
\title{{\sc Incomplete Information and Matching of Likes: A Mechanism Design Approach}\thanks{
We are grateful to Ivan Balbuzanov, Lars Ehlers, Deniz Dizdar, Nicolas
Klein, Shih En Lu, Debasis Mishra, Herv\'{e} Moulin, Szilvia P\'{a}pai, Umutcan Salman,
Christopher Sandmann, Arunava Sen, as well as various seminar and conference
audiences for helpful comments.}}
\author{Dinko Dimitrov\thanks{Saarland University, Germany; email: dinko.dimitrov@mx.uni-saarland.de}\hspace{0.17cm} and Dipjyoti Majumdar\thanks{Concordia University, Canada; email: dipjyoti.majumdar@concordia.ca}}
\date{\today}
\maketitle

\begin{abstract}
We study the implementability of stable matchings in a two-sided market model with one-sided incomplete information. Firms' types are publicly known, whereas workers' types are private information. A mechanism generates a matching and additional announcements to the firms at each report profile of workers' types. When agents' preferences are increasing in the types of their matched partner, we show that the assortative matching mechanism which publicly announces the entire set of reported types is incentive compatible. Furthermore, any mechanism that limits information disclosure to firms’ lower contour sets of reported types remains incentive compatible. However, when information is incomplete on both sides of the market, assortative matching is no longer implementable.


\end{abstract}

\section{Introduction}

Since the seminal paper by \citet{GS}, there has been an extensive literature on two-sided matching that examines markets characterized by heterogeneity on both sides. In the Gale-Shapley model, agents on each side of the market are endowed with exogenous preferences over agents on the opposite side, with stability being the central concept. A matching is said to be stable if there exists no pair of agents - one from each side of the market - who would both prefer to be matched with each other rather than with their assigned partners. Much of the subsequent literature adopts the assumption that agents’ characteristics are publicly known (see, e.g., \cite{Roth2}).

However, it is natural to expect that at least some agents’ characteristics are private information. For example, a woman in a marriage market may not have precise information about a man’s attributes, and a hospital may learn the quality of a doctor only after hiring him. Similarly, a firm can observe a worker’s productivity only after the employment relationship has begun.

In the present paper, we adopt a mechanism design approach to study the
implementability of stable matchings in a two-sided market model with one-sided private information. Specifically, we consider a framework
in which firms are matched with workers. Each agent is characterized by a one-dimensional quality parameter (type). Firms' types are commonly known, and agents' preferences are increasing in the types of their matched partner. Workers' types are private information and a firm observes a worker's type only after being matched with her. Our setup is similar to the one in the influential paper by \cite{Liu2}, but differs in that we abstract from monetary transfers. Under complete information, the unique stable matching in this environment is the (\textit{positively}) \textit{assortative matching}, in which agents of higher types on each side of the market are matched with one another. With incomplete information, however, and depending on the notion of stability adopted, many more outcomes can be stable. We therefore ask whether there exist mechanisms that truthfully implement the complete information stable matching for every possible realization of workers' types, given that firms' types are commonly known. Addressing this question requires, first, defining an appropriate notion of stability under incomplete information and, second, designing mechanisms that induce workers to report their private information truthfully in equilibrium.

The first issue pertains to the formulation of an appropriate notion of blocking. In our setting, a \textit{matching state} consists of a matching allocation and an information structure (containing an assignment of types and a profile of information sets). A matching state is blocked by a firm
and a worker at a given type assignment if (i) the worker prefers the firm
to her current match and (ii) the firm \textit{perceives} the worker to be of a higher type than its currently matched partner. To clarify the second point, a firm can form a blocking pair with a worker only if, for \textit{all} type assignments that the firm considers possible, the worker’s type exceeds that of the firm’s current match. We assume that all matchings are individually rational; consequently, a matching state is stable if the matching allocation admits no blocking pair relative to the underlying information structure.

The second issue concerns strategic behavior by workers. When a worker potentially misreports her type, it is necessary to specify the assumptions she makes about the possible strategic behavior of other workers. We consider a setting in which, at the time of reporting, each worker assumes that all other workers report their types truthfully; this assumption is common knowledge. Firms, however, allow for the possibility that more than one worker may have misreported. A worker chooses to misreport her type if two conditions are satisfied: (i) she benefits from the misreport, and (ii) the resulting matching outcome is not blocked, according to the notion of blocking defined in the previous paragraph.

In our model, each firm observes the realized type of the worker to whom it is matched. Nevertheless, as shown in Example 1 below, observing the matched partner’s realized type alone is insufficient to deter misreporting, and mechanisms that implement positive assortative matching may fail to exist. To prevent profitable misreporting, firms must be able to identify potential misreports by their matched workers. Moreover, such identification must induce \textit{nontrivial belief updating} by firms regarding the distribution of types among the remaining workers.

In general, the existence of such mechanisms depends critically on what is observable. In the absence of transfers, the mechanism must include a device that conveys to firms whether some workers may have misreported their types. In our model, each worker submits a report (a type) to the mechanism designer. We propose a mechanism that, for every profile of reports, announces the positively assortative matching with respect to the reported types; if multiple workers submit identical reports, ties are broken according to a pre-specified order. In addition, the mechanism publicly announces the \textit{entire} set of reported types to all firms. This disclosure allows a firm to detect a potential misreport by its matched worker. However, as we demonstrate in Section 4, the \textit{mere} detection of a possible misreport does not ensure that the resulting matching state is blocked. To initiate a block, a firm must be certain - given its information set - that the worker involved in the potential blocking pair has a type strictly higher than that of its current match. More precisely, upon detecting a misreport by its matched worker, a firm must be able to infer the \textit{position} of the worker’s true type relative to the profile of reported types. Because workers only report types that are weakly higher than their true types, the report profile enables a firm to identify workers whose true types must be lower than the type of its matched partner, thereby ruling out certain type assignments. However, unlike \cite{Liu2}, this inference does not generate a recursive elimination of possible type assignments. Instead, we employ an induction argument based on the \textit{position} of a worker’s type within a given type assignment. Specifically, we first show that a worker with assigned type $t$ cannot successfully misreport - that is, misreport and avoid having the resulting matching blocked - at any type assignment in which $t$ is ranked second. This constitutes Step 1 in the proof of our main result, Theorem \ref{Theorem1}. We then establish that a worker of type $t$ cannot successfully misreport at any type assignment in which her type is ranked third, taking Step 1 as given. Since the number of positions is finite, this inductive process terminates, leading us to conclude that the proposed mechanism implements the complete information stable matching allocation - namely, the positively assortative matching.

A natural question arising from the definition of the above mechanism concerns the structure of the additional announcements. In particular, we ask whether an assortative matching mechanism that conveys \textit{less information} to firms can remain non-manipulable. We investigate this issue in Section 4, where it is shown to be closely related to the following observation. Consider a firm $i$, and suppose it has some information about the reported types of workers matched to firms with types higher than that of firm $i$. Such information is immaterial for firm $i$'s search for potential blocking partners, due to the assortative nature of the induced matching. Specifically, any worker $j$ whose reported type exceeds the reported type of firm $i$'s matched partner must be matched with a firm of higher type than $i$. Since firms’ types are publicly known, worker $j$ would not be willing to form a blocking pair with firm $i$. Put differently, and drawing on the insights from the proof of Theorem \ref{Theorem1}, it is sufficient for each firm to observe the entire \textit{lower contour set} of reported types. More precisely, when a firm detects a misreport by its matched worker, it only needs to know the \textit{position} of the worker’s true type relative to the sub-profile of reports in its strict lower contour set. By contrast, if a mechanism withholds information about the reports in firms’ strict lower contour sets, implementation of the complete information stable matching allocation may fail. Finally, in Section 5 we show that when there is private information on both sides of the market, even the assortative matching mechanism that reveals the entire set of reported types to each firm is not incentive compatible.

\subsection{Two Examples}

To motivate our results and demonstrate how the proposed mechanisms operate, we consider two illustrative examples. In both examples, the mechanism generates the assortative matching with respect to the workers’ reported types; moreover, this fact is common knowledge and the resulting matching is publicly observable. In Example 1, no additional announcements are made. By contrast, in Example 2, the mechanism additionally announces to each firm the entire set of reported types.

Throughout, we denote a generic type assignment by the vector $%
(w(j))_{j\in J}$ where $J$ is the set of workers and each $w(j)$ takes
values in a finite set $T$. The set $T$ is endowed with a complete linear
order $>$ which represents the common preference ordering over workers'
types. In each example, we fix a particular worker after the realization of
her type and examine, from her perspective, whether it is profitable to misreport her type to the corresponding mechanism. This evaluation is conducted at a \textit{given realized workers' type assignment}.

In Example 1, we show that the worker can \textit{successfully} misreport, in the sense that the matching outcome resulting from the misreport is stable. This establishes that the corresponding mechanism is \textit{not} incentive compatible. By contrast, the mechanism employed in Example 2 allows us to conclude that the same misreport by the worker, at the same type assignment, is unsuccessful, since the matching outcome induced by the misreport is blocked.

\begin{example}\label{1}
{\rm
Let $J=\{j_{1},j_{2},j_{3}\}$ be the set of
workers and $I=\{i_{1},i_{2},i_{3}\}$ be the set of firms. Firms' types are
publicly known with $i_{1}$ being of the highest type, $i_{2}$ of the second
highest, and $i_{3}$ of the lowest type. Consider a particular realization
of workers' types $w=(w(j_{1}),w(j_{2}),w(j_{3}))=\left(
t_{1},t_{2},t_{3}\right) $ with $t_{1}>t_{2}>t_{3}$, as well as a specific
misreport by $j_{2}$ at $w$ via some $t^{\prime }>t_{1}$. The resulting
assortative matching would be $\mu
^{A}=\{(i_{1},j_{2}),(i_{2},j_{1}),(i_{3},j_{3})\}$ and notice that, by
firms' types being publicly known, neither $j_{2}$ nor $i_{3}$ can belong to
a blocking pair. Given that there is no information other than the matching $%
\mu ^{A}$ that is observed, firm $i_{1}$ only gets to know $\mu ^{A}$ and
the realized type $w(j_{2})=t_{2}$ of worker $j_{2}$. Thus, firm $i_{1}$
cannot rule out as impossible a type assignment $w^{\prime }$ where $%
w^{\prime }(j_{2})=w(j_{2})>w^{\prime }(j)$ for all $j\neq j_{2}$. In words,
the set of type assignments that $i_{1}$ considers as possible (firm $i_{1}$'s information set), given $w(j_{2})$, includes a type assignment where $j_{2}$ has the highest assigned type. Consequently, $i_{1}$ cannot be a part of any blocking. A similar argument implies that firm $i_{2}$ cannot be a part of any blocking, either. Since $j_{2}$ can make this reasoning as well as firm $%
i_{1}$, worker $j_{2}$ can successfully misreport at the type assignment $w$%
. Thus, the mechanism is not incentive compatible.
}
\end{example}

\begin{example}\label{2}
{\rm
Consider the same sets of workers and firms as
in Example 1, as well as the same publicly known firms' types. The
particular realization of workers' types is again $w=\left(
t_{1},t_{2},t_{3}\right) $ with $t_{1}>t_{2}>t_{3}$ and worker $j_{2}$
misreporting at $w$ via $t^{\prime }>t_{1}$. The problem of incentives in
Example 1\ emanates from the fact that, based on what firm $i_{1}$ knows
after the matching, there is no refinement of $i_{1}$'s information set,
i.e., there is no updating of beliefs for firm $i_{1}$. We reiterate here,
in our model, information sets represent the beliefs of the firms. In the
present example, we introduce an additional source of information for the
agents in the form of an announcement by the mechanism of the set of
reported types to each firm.\footnote{%
Recall that it is common knowledge that the mechanism generates the positively assortative matching with respect to the reported types. Hence, the additional announcement of the \textit{entire} set of reported types makes the reported type \textit{profile} known to all firms.} We make the assumption that a firm $i$, after detecting a misreport
(that is, a mismatch between the report made by its partner $j$ and $j$'s
realized type $t$), updates its information set to include only type
assignments where there exists some $j^{\prime }\neq j$ whose assigned type
is higher than $t$ (\textit{nontrivial updating}). Consider now the report
profile $r=\left( r(j_{1}),r(j_{2}),r(j_{3})\right) =(t_{1},t^{\prime
},t_{3})$ and observe that, under the proposed mechanism, firm $i_{1}$ will
get to know the entire set of reports and thus conclude that $t^{\prime }$
was reported by $j_{2}$, $t_{1}$ by $j_{1}$, and $t_{3}$\ by $j_{3}$. Since
workers report only upwards to the mechanism and given that $i_{1}$ observes
the true type ($t_{2}$) of its matched worker ($j_{2}$), we have from $%
r(j_{3})=t_{3}<t_{2}$ that firm $i_{1}$ would only consider those type
assignments $w^{\prime }$ as possible where $w^{\prime
}(j_{3})<t_{2}=w(j_{2})$. Denoting firm $i_{1}$'s information set as $\Pi
_{i_{1}}$, we then have%
\begin{equation*}
\Pi _{i_{1}}=\{w^{\prime }\mid w^{\prime }(j_{3})<t_{2}=w(j_{2})\}.
\end{equation*}%
So, the only issue is whether (i) $w^{\prime }(j_{1})<t_{2}=w(j_{2})$ or
(ii) $w^{\prime }(j_{1})>t_{2}=w(j_{2})$. Notice that $i_{1}$ can conclude
that its matched worker $j_{2}$ has misreported to the mechanism due to $%
t^{\prime }>t_{2}$ with both $t^{\prime }$ and $t_{2}$ being observed by $%
i_{1}$. Then, under nontrivial updating, there must exist at least one
worker $j$ whose assigned type is higher than $w(j_{2})=t_{2}$. Therefore, it
must be the case that (ii) holds, i.e., for all $w^{\prime }\in \Pi _{i_{1}}$%
, $w^{\prime }(j_{1})>t_{2}=w(j_{2})$. Notice finally that firm $i_{1}$ is
of the highest possible type and this is publicly known. Thus, when the
profile of reports is $r=(t_{1},t^{\prime },t_{3})$, the pair $(i_{1},j_{1})$
would block the resulting assortative matching. Our incentive requirement is
that each worker, at the time of sending her report, assumes every other
worker is reporting truthfully. So, worker $j_{2}$ misreporting at type
assignment $w$ via $t^{\prime }$ would think that workers $j_{1}$ and $j_{2}$
would report $r(j_{1})=t_{1}$ and $r(j_{3})=t_{3}$, respectively. Given that
the resulting matching state is not stable, worker $j_{2}$ (who can make the
above reasoning as well as firm $i_{1}$) cannot successfully misreport at $w$
via $t^{\prime }$. The argument presented above applies to the type assignment $w$ in which the type of $j_{2}$ is \textquotedblleft second-ranked\textquotedblright. Our main result extends this reasoning to all type assignments and to any rank that a worker’s assigned type may have.
}
\end{example}

\subsection{Related Literature}

The recent literature on stable matching with incomplete information was initiated by the pioneering work of \citet{Liu2}. In that paper, the authors analyze how firms form beliefs (or, in our terminology, information sets) about the possible distributions of worker types, conditional on a matching being not blocked. This rules out certain distributions of worker types from a firm's information set given that the remaining matching is stable. In their model, the critical notion is a matching outcome - that is, a matching and a type distribution. The inferences that a firm may draw regarding possible type distributions may lead to further inferences via a procedure of iterated elimination of blocked matching outcomes, similar in spirit to the concept of rationalizability (\citet{Bernheim}, \citet{Pearce}). The set of incomplete information stable matching outcomes in the paper by \citet{Liu2} is generally a superset of the set of complete information stable outcomes. Framed in our context, the complete information stable outcome would be the complete information assortative matching. Among many other results, the paper by \cite{Liu2} also provides conditions under which incomplete information stability coincides with complete information stability (Proposition 6 in \cite{Liu2}). However, a critical requirement in that paper is the presence of monetary transfers, without which the set of
incomplete information stable outcomes is much larger. The notion of
stability in \cite{Liu2} is \textquotedblleft belief free\textquotedblright , as in our setup. This is in contrast to the one-sided incomplete information model with non-transferable utility in \citet{Bik1}, where the focus is on the existence of Bayesian stable matching outcomes which are not blocked with respect to a prior belief. \citet{Bik1} also presents a centralized matching mechanism and studies its ex-post incentive compatibility. Although similar in spirit, the format of the mechanism and the analysis there are different from ours. More recently, \citet{Chen1} extend the model in \citet{Liu2} to consider incomplete information on both sides of the market. In a model with two-sided uncertainty and no transfers, \citet{Laz} study a different notion of blocking: a pair of agents will block a matching if there is a positive probability of each agent doing better. Compared to the papers mentioned above, our notion of stability under incomplete information is more permissive and akin to the notion of \textit{naive} blocking in \citet{Chen1}. In fact, virtually any matching allocation is stable in our setup. However, when suitably incorporated into the definition of incentive compatibility, our stability notion allows for a detailed and precise study of the existence of mechanisms implementing the complete information stable matching allocation.

With interdependent preferences and incomplete information on one side of
the market, \citet{chak} study stable matching mechanisms that
elicit information truthfully and are immune to rematching by some
participants, based on updated posterior beliefs. Their analysis focuses on student-college matching, where students have complete information about colleges, but colleges receive noisy signals about students. Importantly, a college’s estimate of a student’s quality depends on the signals received by other colleges. \citet{chak} show that when the entire matching is publicly observed, stable mechanisms generally do not exist. In contrast, we incorporate stability directly into our notion of incentive compatibility and demonstrate that firms’ observations of the reports within their corresponding lower contour sets are crucial for the existence of incentive compatible mechanisms. When a mechanism hides the reports from these sets, existence is not guaranteed. In this sense, our setting and analysis are largely orthogonal to the cited work. In a related but very different two-sided matching model with interdependent preferences, \cite{Sen1} show that ex-post stability and ex-post incentive compatibility of matching rules are not mutually consistent in their setting.

Earlier work on stability in two-sided matching markets with incomplete information has largely focused on the standard Gale-Shapley setting (\cite{GS}), where agents’ preferences are exogenous. In this framework, agents are uncertain about the preferences of others but are fully informed about the quality of each potential match. \cite{Roth1} shows that no stable matching mechanism exists in which truthfully revealing preferences is a dominant strategy for all agents. Building on this, \citet{Ehlers} consider a model with two-sided incomplete information and demonstrate that an ordinal Bayesian incentive compatible mechanism exists if and only if there is exactly one stable matching in every state of the world. Similar negative results are reported in \cite{Maj1}. In a one-to-one matching model with transfers, \cite{Yenmez} studies the compatibility of incentive constraints with the existence of stable, efficient, and budget-balanced mechanisms.

The remainder of the paper is organized as follows. Section 2 introduces the basic model and necessary preliminaries. Section 3 presents the proposed incentive compatible assortative matching mechanism, while Section 4 studies mechanisms that convey less information to firms. Section 5 concludes by examining the case of two-sided incomplete information. The Appendix contains the proofs of the main results.

\section{Preliminaries}

We consider a model where there are two disjoint sets of agents $%
I=\{i_{1},\cdots ,i_{n}\}$ and $J=\{j_{1},\cdots ,j_{n}\}$. The two sets
have the same number $n$ of agents, $n<\infty $. For the ease of
exposition, we will refer to the members of the set $I$ as firms and to the
members of the set $J$ as workers.

A \textit{matching} $\mu :I\cup J\rightarrow I\cup J$ is a bijection such
that (i) for any $i\in I$, $\mu (i)\in J\cup \{i\}$, (ii) for any $j\in J$, $\mu
(j)\in I\cup \{j\}$, and (iii) for any $i\in I$, $\mu (i)=j\Leftrightarrow \mu
(j)=i$. The interpretation of $\mu (\ell )=\ell $ for $\ell \in I\cup J$ is
that agent $\ell $ is unmatched. We denote by $\mathcal{M}$ the set of all
matchings.

\vspace{0.3cm}
\noindent
\textsc{Types and Preferences:} In the standard Gale-Shapley
model each agent has a strict preference ordering over the agents on the
other side of the market (and the possibility of remaining unmatched). In
order to incorporate incomplete information in the model, following the
seminal paper by \citet{Liu2}, we let each agent's productivity
parameter be described by the agent's \textit{type}. Denoting by $%
T\subset \Re _{+}$ the finite set of possible workers' types, we write $%
T=\{t_{1},\cdots ,t_{L}\}$ for some integer $L>n$. Without loss of
generality, we assume $t_{1}>t_{2}>\cdots >t_{L}>0$.

In a similar fashion, we denote by $S\subset \Re _{+}$ the finite set of
possible firms' types with $f:I\rightarrow S$ mapping each firm to its type.
Since firms' types are \textit{commonly} \textit{known}, we assume that
these are collected in the set $S=\{s_{1},\cdots ,s_{n}\}$ with $%
s_{1}>s_{2}>\cdots >s_{n}>0$ and $f(i_{\ell })=s_{\ell }$ holding for all $%
\ell \in \{1,\cdots ,n\}$. Each worker $j\in J$ has a strict preference $%
\succ _{j}$ over the set of types of her potential matches - that is, the
set $S$. However, a worker may end up remaining unmatched. As a convention,
we say that if a worker is unmatched, she is matched to a firm of type $0$. We assume further that $\succ_j$
defined over $S\cup \{0\}$ has the following features:

\begin{itemize}
  \item for all $\ell ,k\in \left\{ 1,\cdots ,n\right\} $, $\succ _{j_{\ell
}}=\succ _{j_{k}}$,
  \item for all $\ell \in \{1,\cdots ,n\}$ and all $k,k^{\prime }\in
\{1,\cdots ,n\}$, $s_{k}\succ _{j_{\ell }}s_{k^{\prime }}\Leftrightarrow
s_{k}>s_{k^{\prime }}$.
\end{itemize}

The first condition says that the preference over types for all workers is
the same and the second illustrates the rather standard assumption (see, e.g., \citet{Bik1}) that preferences are increasing in types. In addition, $s_{1}>s_{2}>\cdots
>s_{n}>0$ implies that each worker prefers to be matched to some firm rather than remain unmatched.

Each firm $i\in I$ has a strict preference ordering $\succ_i$ over the set of possible worker types $T$.  Again, as a convention, if a firm remains unmatched, we assume that the firm is matched to a worker of type $0$. Thus, the preference ordering $\succ_i$ of firm $i\in I$ is formally defined over $T\cup \{0\}$ and has the following features:
\begin{itemize}
  \item for all $\ell, k\in \{1,\cdots,n\}$, $\succ_{i_{\ell}}=\succ_{i_k}$,
  \item for all $\ell\in \{1,\cdots, n\}$ and for all $t, v\in T$, $t\succ_{i_{\ell}}v\Leftrightarrow t>v$.
\end{itemize}

As before, these two conditions respectively state that all firms have
identical preferences over workers' types and that these preferences are
increasing in types. In addition, $t_{1}>t_{2}>\cdots >t_{L}>0$ implies that
each firm prefers to be matched to some worker rather than remain unmatched.

The types of the workers are \textit{private information}. We shall call a function $w:J\rightarrow T$ mapping a worker to her type a \textit{type assignment} function, and let ${\clw}=\{w:J \rightarrow T \,\mid \forall j\neq j'\in J,\,\, w(j)\neq w(j')\}\subset T^n$ be the set of all type assignment maps. Note the assumption that no two workers are assigned the same type.

\vspace{ 0.3 cm}
\noindent
\textsc{Information Structures:} The incomplete information
scenario we consider is the following: the type $f(i)$ of each firm $i\in I$
is commonly known; each worker $j\in J$ knows her own type, but the specific
type assignment $w$ is not known to her; if firm $i$ and worker $j$ are
matched, firm $i$ perfectly observes the type of worker $j$ - its matched
partner. Therefore, a type assignment $w\in \mathcal{W}$ and a matching $\mu
\in \mathcal{M}$ induce an \textit{information set} for each agent $k\in I\cup J$. In addition to the matching $\mu$ that is publicly observed, there may be some other \textit{payoff relevant signal}. The finite set of such signals is denoted by $\Lambda$.

Consider now a matching $\mu \in \mathcal{M}$ and a type assignment $w\in
\mathcal{W}$. For a firm $i\in I$, we denote by $\Pi _{i}(w,\mu )\subset
\mathcal{W}$ the subset of type assignments that are consistent with $w$ and
$\mu $; that is, consistent with the fact that firm $i$ observes the
realized type of its matched worker, i.e.,
\[
\Pi_i(w,\mu)=\{w'\in \clw\,\mid w'(\mu(i))=w(\mu(i))\}\subset \clw.
\]

 Likewise, given $w\in \clw$ and $\mu\in \mathcal{M}$, for any worker $j\in J$ we denote by $\Pi_j(w,\mu)\subset \clw$ the subset of type assignments that are consistent with $w$ and $\mu$. In other words, this is the subset of type assignments that worker $j$ perceives as possible once she observes $w(j)$ and $\mu$, i.e., $\Pi_j(w,\mu)=\{w'\in \clw\,\mid w'(j)=w(j)\}\subset \clw$.

 Consider now any agent $k\in I\cup J$. Suppose that agent $k$ observes some additional payoff relevant signal $\lambda_k\in \Lambda$. Such a signal potentially refines $k$'s information set. That is, for any $k\in I\cup J$, type assignment $w$, matching $\mu$, and signal $\lambda_k$,
\begin{equation}\label{eq1}
\Pi_k(w,\mu,\lambda_k) \subseteq \Pi_k(w,\mu),
\end{equation}
 where $\Pi_k(w,\mu,\lambda_k)$ denotes the refinement of agent $k$'s information set. The non-strict subset relation in \eqref{eq1} is due to the fact that we allow for signals to be completely uninformative.

 For any agent  $k\in I\cup J$, type assignment $w\in \clw$, matching $\mu\in \mathcal{M}$, and any vector of payoff relevant signals $\lambda=(\lambda_k)_{k\in I\cup J}\in \Lambda^{2n}$, we refer to $\Pi_k(w,\mu,\lambda)$ as agent $k$'s information set corresponding to $(w,\mu,\lambda)$, and to the collection ${\Pi}=\{\Pi_k()\}_{k\in I\cup J}$ as the \textit{information structure}. Throughout, we assume that the information structure ${\Pi}$ is \textit{commonly known}.

 \vskip 0.3cm
 \noindent
 \textsc{Matching States:}  Given a vector of signals $\lambda\in \Lambda^{2n}$, a market is described by a matching allocation $\mu $, type assignment $w$,
 and the collection of information sets $\Pi (w,\mu ,\lambda )$. Formally, for every $w\in
\mathcal{W}$, $\mu \in \mathcal{M}$, and $\lambda \in \Lambda ^{2n}$, a
\textit{matching state} is the collection $\left( w,\mu ,\Pi (w,\mu ,\lambda
)\right) $.

 \vskip 0.3cm
 \noindent
 \textsc{Individual Rationality:} A matching state is \textit{%
individually rational}, if the payoff for each agent from the matching is at
least as high as the outside option of remaining unmatched. Since firms'
types are commonly known and firms do observe the type of the worker they
are matched with, in the present model there is no difference in individual
rationality between complete and incomplete information settings.

\begin{defn}\label{dfIR} A matching state $\left( w,\mu ,\Pi
(w,\mu ,\lambda )\right)$ is \textbf{individually rational}, if
\begin{itemize}
\item for each firm $i\in I$, $w(\mu (i))>0$, and
\item for each worker $j\in J$, $f(\mu (j))>0$.
\end{itemize}
\end{defn}

Given the preferences over types in our model, for all $\mu \in \mathcal{M}$%
, $w\in \mathcal{W}$, and $\lambda \in \Lambda ^{2n}$, the matching state $%
\left( w,\mu ,\Pi (w,\mu ,\lambda )\right) $ is individually rational.

We now introduce the notion of stable matching. In a complete
information environment, this notion is well established (\cite{GS}). With incomplete information, we are concerned with \textit{stable matching states} and emphasize once again what each agent knows: the types of all the firms, the entire matching, that is which
worker is matched to which firm, and the type of its/her current matched
partner. Moreover, there may be some additional payoff relevant signals. This leads us to the following \textit{naive} notion of blocking.

\begin{defn}\label{dfBLOCKING}
A matching state $\left( w,\mu ,\Pi
(w,\mu ,\lambda )\right) $ is \textbf{blocked} by a
firm-worker pair $(i,j)$, $i\in I$ and $j\in J$,
if
\begin{itemize}
  \item $s_{i}>s_{\mu (j)}$, and
  \item for all $w^{\prime }\in \Pi _{i}(w,\mu ,\lambda _{i})$, $
w^{\prime }(j)>w^{\prime }(\mu (i))$.
\end{itemize}
\end{defn}

Thus, a worker $j$ will agree to form a blocking pair with firm $i$, if the
type of firm $i$ is higher than the type of the firm the worker is matched
to under $\mu $. A firm $i\in I$ will enter into a blocking with a worker $j$%
, if firm $i$ is certain that $j$ has an assigned type that is higher than
the type of the worker it is presently matched to. This idea of blocking is
similar to the corresponding ideas in \cite{Liu2} and \cite{Bik1}.  However, unlike those papers,
there is no iterated elimination of blocked matching states in our definition of stability.

\begin{defn}\label{dfSTABILITY}
A matching state $\big(w,\mu,\Pi(w,\mu,\lambda)\big)$  is \textbf{stable} if it is not blocked by any pair $(i,j)\in I\times J$.
\end{defn}

It is immediate that our notion of stability is permissive. Since types are assumed to be individually rational for the corresponding agents, individual rationality is implicitly incorporated into the above definition. Moreover, a matching state is unlikely to be blocked simply because the second condition in Definition \ref{dfBLOCKING} must hold for all type assignments compatible with the information available to a firm at the \textit{realized} type assignment. Below, we formalize the permissiveness of this concept.

Note that our notion of stability is defined in terms of matching states and therefore depends critically on the coarseness of firms’ information sets. One particular class of matching states that is of interest is the following.

\begin{defn}\label{MI}
A matching state $\left( w,\mu ,\Pi
(w,\mu ,\lambda )\right) $ is \textbf{minimally informative} if for
each firm $i\in I$,
\[
\Pi _{i}(w,\mu ,\lambda _{i})=\left\{ w^{\prime }\in \mathcal{W}\mid
w^{\prime }\left( \mu \left( i\right) \right) =w\left( \mu \left( i\right)
\right) \right\}.
\]
\end{defn}

In other words, in a minimally informative matching state, each firm has very limited knowledge about the possible type assignments of workers, knowing only the type of its matched partner. Importantly, the definition of minimal informativeness does not impose any restrictions on workers’ information sets.

We demonstrate below that there can be many matching states that are
minimally informative and stable. Recall that $i_{n}$ is the firm of the
lowest type, $f(i_{n})=s_{n}$, and that $t_{L}=\min \{T\}$ is the lowest
possible type of a worker.

\begin{prop}\label{prop1}
A minimally informative matching state $\left( w,\mu ,\Pi (w,\mu ,\lambda )\right) $ is stable if
and only if $w(j)=t_{L}$ for some $j\in J$ implies $\mu(i_{n})=j$.
\end{prop}

\begin{proof}
\noindent \textsc{($\Leftarrow $) Part:} Let $w\in \mathcal{W%
}$ be a type assignment such that $w(j_{k})=t_{L}$ for some $j_{k}\in J$ and
let $\mu $ be a matching where $\mu (i_{n})=j_{k}$. We will show that the
(minimally informative) matching state $\left( w,\mu ,\Pi (w,\mu ,\lambda
)\right) $ is stable. Recall that firm $i_{n}$ is of the respective
lowest (and publicly observable) assigned type; that is, it can never be a part of any blocking pair. Now consider any firm $i\neq i_{n}$ and let $\mu (i)=j$. Note that $j\neq j_{k}$ holds and thus, $w(j)>t_{L}$ follows. This implies that for any $j^{\prime}\neq j=\mu (i)$ there exists a type assignment $w^{\prime }\in \Pi_{i}(w,\mu ,\lambda _{i})$ such that $w^{\prime }(\mu (i)=j)>w^{\prime}(j^{\prime })$. Consequently, firm $i$ would not be a part of any blocking pair. Since $i$ was arbitrarily chosen, this completes the argument.

\vskip 0.2cm
\noindent
\textsc{($\Rightarrow $) Part:} Let the matching state $\left( w,\mu ,\Pi
(w,\mu ,\lambda )\right) $ be minimally informative and let $j_{k}\in J$ be
such that $w(j_{k})=t_{L}$. We claim that if $\left( w,\mu ,\Pi
(w,\mu ,\lambda )\right) $ is stable then $\mu(i_n)=j_k$. Suppose the claim is false, i.e., the matching $\mu $ is such that $\mu (i_{n})\neq j_{k}$. We will establish that $\left( w,\mu ,\Pi (w,\mu
,\lambda )\right) $ is not stable. Let $\mu (i_{n})=j_{m}$ and $\mu (i_{\ell
})=j_{k}$; that is, the firm with the lowest possible type is matched to worker $j_m$.  Note that $s_{\ell }=f(i_{\ell })>s_{n}=f(i_{n})$. Thus, worker $j_m$ prefers to be matched to firm $i_{\ell}$ over her current match $i_n$.  Moreover, for any type assignment $w^{\prime }\in \Pi _{i_{\ell }}(w,\mu ,\lambda
_{i_{\ell }})$, $w^{\prime }(j_{m})>w^{\prime }(j_{k})=t_{L}$. This last part follows since no two workers have the same assigned type and worker $j_k$ has the lowest assigned type $t_L$. Consequently, $(i_{\ell },j_{m})$ would block the matching state $\left( w,\mu ,\Pi (w,\mu,\lambda )\right)$.
\end{proof}

One important matching that we are interested in is the assortative matching.

\begin{defn}
A matching $\mu$ is \textbf{assortative} at a type assignment $w\in \clw$ if for each $i, i'\in I$,
\[
[f(i) > f(i')]\Leftrightarrow [w(\mu(i)) > w(\mu(i'))].
\]
A matching $\mu $ is assortative, if for every $w\in
\mathcal{W}$, $\mu $ is assortative at $w$. We
denote the assortative matching by $\mu ^{A}$.
\end{defn}

Proposition \ref{prop1} immediately leads to the following observation.

\begin{obs}\label{obs1}
{\rm Let $\mu ^{A}\in \mathcal{M}$ be
the assortative matching and let $w\in \mathcal{W}$ be a realized type
assignment. Then any minimally informative matching state $\left( w,\mu
^{A},\Pi (w,\mu ^{A},\lambda )\right) $ is stable. In other words, at any type assignment $w\in \mathcal{W}$, the complete information stable matching is contained in the set of matching states that are minimally informative and stable.}
\end{obs}

\subsection{Matching Mechanisms}

In a complete information environment where the types of both workers and firms are commonly known and preferences are increasing in types, the unique stable matching is assortative. In our setup, the definition of a blocking pair is demanding, which makes the corresponding notion of stability permissive: it allows matchings that may be stable under incomplete information but would not necessarily be stable under complete information. By contrast, \citet{Liu2} adopt a more stringent stability concept and identify conditions under which stability under complete information coincides with stability under incomplete information.

Given that workers’ types are privately known and, at the interim stage, are observed by the firms to which they are matched, one natural approach to studying matching under incomplete information is to focus on \textit{mechanisms} that implement particular matching rules.

A mechanism in our model asks each worker to report her type and for each profile of \textit{reported} types, specifies  a matching together with an
\textit{additional announcement} for each agent $k\in I\cup J$ . Notice that, in general, it may happen that two or more workers report the same type. In that case, one needs to break ties. A tie-breaking rule $\tau $ is simply an ordering of the workers in $J$. In other words, $\tau :\left\{ 1,\cdots ,n\right\} \rightarrow \left\{ 1,\cdots ,n\right\} $ is a bijection defining a precedence relation $\tau (1)\vartriangleright \cdots \vartriangleright \tau (n)$. We further assume that the additional announcements made by the mechanism are simply subsets of the set $T$ of workers' types.

\begin{defn}\label{dfmatchingmech}
Fix a tie-breaking rule $\tau$. A \textbf{matching mechanism} is a pair $\langle \sigma _{\tau},\alpha \rangle$ where $\sigma _{\tau }:T^{n}\rightarrow \mathcal{
M}$  is a matching rule and $\alpha =\left( \alpha
_{k}:T^{n}\rightarrow 2^{T}\right) _{k\in I\cup J}$ is a collection of
announcement rules.
\end{defn}

Therefore, for each profile of reports $r=(r(j))_{j\in J}$, a mechanism selects a matching $\sigma_{\tau}(r)$ and announces a subset $\alpha_k(r)$ of the set of workers' types to each agent $k\in I\cup J$. Following the standard tradition in mechanism design, we assume that the matching rule and the announcement rules are commonly known and that the designer fully commits to the mechanism. For the rest of our analysis of the one-sided incomplete information model, announcements made to the workers are redundant. Consequently, except for the concluding section where we elaborate on the two-sided incomplete information case, we will assume that for each $j\in J$, and for each report profile $r\in T^n$, $\alpha_j(r)=\emptyset$. Further, we denote by $\sigma _{\tau }(r;k)$ the matching partner of agent $k\in I\cup J$ under the matching rule $\sigma_{\tau }$ at $r\in T^{n}$; that is, if $\sigma _{\tau }(r)=\mu $, then $\sigma _{\tau }(r;k)\equiv \mu (k)$. The specific matching rule we are interested in is the \textit{assortative} rule that picks, for each profile of reports, the matching that is assortative with respect to the \textit{reported} types.

\begin{defn}\label{dfasmatrule}
Given a tie-breaking rule $\tau $, the matching rule $\sigma _{\tau }:T^{n}\rightarrow \mathcal{M}$ is \textbf{assortative} with respect to the reported types, if for
each $r\in T^{n}$ and each $i,i^{\prime }\in I$,
\[
\lbrack f(i)>f(i^{\prime })]\Rightarrow \lbrack r(\sigma _{\tau }(r;i))\geq
r(\sigma _{\tau }(r;i^{\prime }))]
\]
 with ties broken according to $\tau$.
\end{defn}

In  other words, the assortative matching rule operates as follows. For any two workers, the worker with the higher report is assigned to the firm with the higher type. If the two workers submit the same report, the worker who takes precedence in the ordering $\tau$ is assigned to the firm with the higher type. Given a tie-breaking rule $\tau $, we denote the corresponding assortative matching rule by $\sigma _{\tau }^{A}$.

Since workers' types are private information, the workers are aware that
they can manipulate the outcome of a matching mechanism by misreporting
their types. However, unlike standard mechanism design problems, the issue
of misreporting is more nuanced. Whenever a matching allocation $\mu $
realizes at a type assignment $w$, each firm $i$ gets to know the type $w(\mu (i))$ of its matched worker $\mu (i)$. A misreport by a worker will lead to a change in the matching prescribed by the mechanism. Following such a misreport by a worker, the firm may find that the observed type of its matched worker is \textit{inconsistent} with the information it has. To be more specific, consider a mechanism $\langle \sigma_{\tau},\alpha \rangle $, some worker $j$ of type $w(j)$, and suppose that worker $j$ misreports via some $r(j)\in T$, $r(j)\neq w(j)$. Given the reports $r\in T^{n}$, let $\mu=\sigma_{\tau}(r)$ denote the resulting matching where worker $j$ is matched to firm $i$. Also let $\lambda_i=\alpha_i(r)$ be the announcement made by the mechanism to firm $i$. Recall that  firm $i$ gets to observe the type $w(j)=w(\mu (i))$ of worker $j$. It may happen that there exists some other worker $j^{\prime }\neq j$ such that for all type assignments $w^{\prime }$ that firm $i$ considers as possible, i.e., for all $w^{\prime }\in \Pi _{i}(w,\mu ,\lambda _{i})$, $w^{\prime
}(j^{\prime })>w^{\prime }(j)=w(j)$. If it is also the case that $f(i)>f(\mu(j^{\prime }))$, then $i$ and $j^{\prime }$ would form a blocking pair for the matching state resulting from the misreport. If the matching state resulting from a misreport is blocked, we call the misreport \textit{unsuccessful}. Otherwise, it is successful. This is the idea behind our notion of \textit{successful manipulability}.

In order to formalize the idea, we introduce a few notations. Given a type
assignment $w\in \mathcal{W}$ and a report $r(j)\in T$ by worker $j$, we
denote by $(w_{-j},r(j))$ the profile of types where the type $w(j)$ is
replaced by $r(j)$. It is important to note that $(w_{-j},r(j))\in T^{n}$
and may not belong to $\mathcal{W}$, in the sense that $r(j)$ may coincide
with some $w(j^{\prime })$ in the vector $w_{-j}$. Since in a type
assignment no two workers are assigned the same type, such a profile $%
(w_{-j},r(j))\in T^{n}$ may not be a valid type assignment. Likewise, for a
report profile $r\in T^{n}$, we let $r_{-j}$ denote the profile of reports
by all workers other than $j$. Given a type assignment $w\in \mathcal{W}$, a
report profile $r\in T^{n}$, and a matching mechanism $\langle \sigma _{\tau
},\alpha \rangle $, we denote by $\Pi _{i}(w,\sigma _{\tau }(r),\alpha
_{i}(r))$ the information set for firm $i\in I$ at $w\in \mathcal{W}$ and $%
r\in T^{n}$ under the corresponding mechanism.

We now proceed with the formal definition of successful manipulability.
Since the types of the firms are commonly known, manipulability only applies
to workers.

\begin{defn}\label{defsm}
A matching mechanism $\langle \sigma_{\tau },\alpha \rangle $ is \textbf{successfully manipulable} by a
worker $j\in J$ at a type assignment $w\in \mathcal{W}$
via report $r(j)=t\in T$ if
 \begin{itemize}
   \item $f(\sigma _{\tau }(\left( w_{-j},t\right) ;j))>f(\sigma _{\tau
}(w;j))$, and
   \item the matching state $\left( w,\sigma _{\tau }(w_{-j},t),\Pi
(w,\sigma _{\tau }(w_{-j},t),\alpha (w_{-j},t)\right) $ is
stable.
 \end{itemize}
\end{defn}

The first condition says that if a worker $j$ misreports at a type
assignment $w$, it must lead to a matching where worker $j$ is matched to a
strictly higher ranked firm. The second condition is the critical part of
our successful manipulability requirement. The manipulation must be
sustainable - that is, the manipulation must not lead to a matching state
that is blocked by some firm-worker pair.

\begin{defn}\label{dfIC}
A matching mechanism $\langle \sigma_{\tau },\alpha \rangle $ is \textbf{incentive compatible} if there
does not exist a worker $j\in J$ who can successfully manipulate the mechanism at some type assignment $w\in \mathcal{W}$ under the assumption that all the other workers are reporting truthfully. A matching rule $\sigma ^{\prime }_{\tau}:T^{n}\rightarrow \mathcal{M}$ is \textbf{implementable}, if there exists an incentive compatible matching mechanism $\langle \sigma _{\tau },\alpha \rangle $  with  $\sigma _{\tau}=\sigma _{\tau }^{\prime }$.
\end{defn}

\begin{obs}\label{obs2}
{\rm Note that, according to our notion of
successful manipulability, when some worker $j\in J$ contemplates
misreporting at some type assignment $w\in \mathcal{W}$ via some $r(j)\neq
w(j)$, she is assuming that every other worker is reporting truthfully.
Consequently, at the assignment $w$ and for $j$'s report $r(j)$, the
relevant matching is $\sigma _{\tau }(w_{-j},r(j))$ and the relevant
collection of announcements is $\alpha (w_{-j},r(j))=\left( \alpha
_{i}(w_{-j},r(j))\right) _{i\in I}$. Since worker $j$ is not assumed to
observe $w$, she must verify the possibility of successful manipulation for
all possible type assignments $w^{\prime }\in \mathcal{W}$ where her type is
$w^{\prime }(j)=w(j)$.}
\end{obs}

Suppose now that the mechanism reveals a worker’s report to the firm to which she is matched. Recall that, in our setup, each firm observes the true type of its matched worker. If the revealed report of the firm’s matched partner differs from the worker’s observed true type, the firm concludes that the worker has misreported. Such an event leads to a refinement of the firm’s information set, which represents its \textit{beliefs} regarding the possible assignments of workers’ types. The notion of nontrivial updating, introduced below, specifies how firms refine their information sets upon detecting a misreport by their matched workers.

To formalize the idea, fix an assortative matching mechanism $\langle
\sigma _{\tau }^{A},\alpha \rangle $, a type assignment $w\in \mathcal{W}$,
and a report profile $r\in T^{n}$. Suppose that a worker $j\in J$ misreports
at $w$ via $r(j)\neq w(j)$. In the setup we consider, a meaningful misreport would imply $r(j)>w(j)$.  Notice that worker $j$ would prefer her match under $\sigma _{\tau }^{A}\left( w_{-j},r(j)\right) $ over her match under $\sigma _{\tau }^{A}\left( w\right) $ only if there is some other worker $j^{\prime }$ with $r(j)\geq w(j^{\prime })>w(j)$; otherwise, worker $j$ would be matched to the same firm under the two matchings. Define then the set
\begin{equation}\label{as1eq1}
\mathcal{W}(w(j),r(j))=\{w^{\prime }\in \mathcal{W}\mid w^{\prime }(j)=w(j)%
\text{ and }\exists \,j^{\prime }\neq j\text{ s.t. }r(j)\geq w^{\prime
}(j^{\prime })>w(j)\}
\end{equation}

\noindent
of type assignments, where the misreport $r(j)$ by worker $j$ of assigned type $w(j)$ could be payoff relevant for her.\footnote{ We would like to stress here that \textquotedblleft misreport\textquotedblright\ does not necessary mean \textquotedblleft
successful manipulation\textquotedblright.}

Consider next a firm $i$ and a report profile $r$. We say that the mechanism $\langle
\sigma _{\tau }^{A},\alpha \rangle $ \textit{reveals} the report of the matched worker to firm $i\in I$ if $r(\sigma_{\tau }^{A}(r;i))\in\alpha _{i}(r)$. Suppose in addition that $r(\sigma _{i}^{A}(r))>w(\sigma _{i}^{A}(r))$ holds; that is, firm $i$ becomes aware that  its matched worker has misreported upwards to the mechanism. In such a case, \textit{nontrivial updating} eliminates from firm $i$'s information set $\Pi _{i}(w,\sigma
_{\tau }^{A}(r),\alpha _{i}(r))$ type assignments which are not payoff
relevant for the misreporting worker. Observe that there are two cases where
it is not important how a firm updates its information: when this firm is $%
i_{n}$ or when a firm $i\in I$ is matched to a worker of the lowest possible
type ($t_{L}$). In the first case (due to firms' types being publicly known)
no worker is willing to form a blocking pair with $i_{n}$, while in the
second case firm $i$ can immediately conclude that forming a blocking pair
with any worker $j^{\prime }\neq j$ is worthy for it. Consequently, we have
the following formal definition of nontrivial updating.

\begin{assumption}\label{NT}
\textbf{} Fix an assortative matching mechanism $\langle \sigma _{\tau }^{A},\alpha \rangle$. We say that the information structure ${\Pi}=\{\Pi_i()\}_{i\in I}$  satisfies \textbf{nontrivial updating} with respect to $\langle \sigma _{\tau }^{A},\alpha \rangle$ if for any type assignment $w\in \mathcal{W}$, any report profile $r\in T^{n}$, and any firm $i\in I\setminus \left\{i_{n}\right\} $, the following holds: $r(\sigma_{\tau }^{A}(r;i))\in\alpha _{i}(r)$ and $r(\sigma_{\tau }^{A}(r;i)) >w(\sigma _{\tau }^{A}(r;i)))>t_{L}$ implies \[\Pi_i(w,\sigma^A_{\tau}(r),\alpha_i(r))\subseteq  \mathcal{W}(w(\sigma^A_{\tau}(r;i),r(\sigma _{\tau }^{A}(r;i))).\]
\end{assumption}
The nontrivial updating assumption can be seen as a minimal rationality requirement imposed on the information structure \textit{associated} with an assortative matching mechanism. It says that each firm $i\neq i_n$ takes into account any information that is ``payoff relevant'' when updating its ``beliefs'' regarding the possible assignment of types of the workers. As already explained above, firm $i_n$ cannot be part of any blocking pair and thus, no restriction is needed on that firm's belief updating. Likewise, if at some type assignment $w$, the assigned type of some worker $j$ is $w(j)=t_L$, then, as Proposition \ref{prop1} demonstrates,  any matching $\mu$ where $\mu(j)\neq i_n$ will be blocked. Hence, no restriction is needed on belief updating in this case, either. Note that when firm $i$ infers
that the report $r(\sigma _{\tau }^{A}(r;i))$ of its matched worker and the
true type of this worker matches the report, that is, $w(\sigma _{\tau
}^{A}(r;i))=r(\sigma _{\tau }^{A}(r;i))$, Assumption \ref{NT} does not
impose any restriction on the updating of beliefs. In addition, if at a
report profile $r\in T^{n}$ firm $i$ is not able to conclude that $r(\sigma
_{\tau }^{A}(r;i))>w(\sigma _{\tau }^{A}(r;i))$ is the case, then nontrivial updating does not impose any restriction, either.

 It is worth mentioning that not every assortative matching mechanism with the associated information structure satisfying nontrivial updating, is incentive compatible. In particular, the implementability of the assortative matching rule will crucially depend on the structure of the additional announcements.  As an example, consider the assortative matching mechanism $\langle\sigma _{\tau }^{A},\alpha ^{\emptyset }\rangle $ where for each firm $i\in I$ and for each report profile $r\in T^n$, $\alpha_i^{\emptyset}(r)=\emptyset$. Notice that for such a mechanism, the associated information structure vacuously satisfies nontrivial updating  (since the first condition in the ``if'' part of Assumption \ref{NT} is not satisfied). Clearly then, such \textit{empty}
announcements will not refine the information sets of the firms for any
profile of reports resulting in all matching states being \textit{minimally
informative}.  Proposition \ref{prop1} and Observation \ref{obs1} then tell us that all such
matching states will be stable, and thus, as explicitly demonstrated by Example \ref{1} in the Introduction, the assortative matching rule $\sigma _{\tau }^{A}$ will not be implementable in this case. The importance of nontrivial updating for assortative matching mechanisms with \textit{nonempty} announcements being incentive compatible is
illustrated in  the next section.

\section{Informative Public Announcements}

The empty signal fails to work because it conveys no information about a worker’s report to the firm to which the worker is matched; as a result, the firm cannot detect a misreport. In what follows, we introduce a specific class of announcement rules that enables such verification of misreports. We say that a collection of announcement rules $\alpha=(\alpha_i)_{i\in I}$ is \textit{public} if for all $i,k\in I$, $\alpha_i=\alpha_k$. In other words, the announcement rule is the same for each firm. In particular, the empty announcement rule as defined in the last paragraph of the previous section is public.

We further say that the collection $\alpha^I=(\alpha^I_i)_{i\in I}$ of public  announcement rules is \textit{informative}, if it is conveys the \textit{entire set} of reported types to every firm. That is, for each $i\in I$ and $r\in T^n$,
\[
\alpha _{i}^{I}(r)=\{t\in T\mid t=r(j)\text{ for some }j\in J\}.
\]
 Consider now the assortative matching rule $\sigma^A_{\tau}$. We call $\langle \sigma _{\tau }^{A},\alpha ^{I}\rangle $ the \textit{assortative matching with informative public announcement (AMIPA)} mechanism. Note that under the AMIPA mechanism the following information becomes publicly known: the assortative matching rule and its outcome at each report profile, as well as the entire set of reported types. Hence, this information leads to the reported type \textit{profile} becoming publicly known: a firm will be able to decipher the identity of the worker who has made a specific report. Since in our setup the firms observe the true type of the worker they are correspondingly matched to, each firm will also be able to detect a possible misreport of its matched worker.

Notice that under the AMIPA mechanism $\langle \sigma^A_{\tau}, \alpha^I \rangle $, we have $r(\sigma^A_{\tau}(r;i))\in\alpha^I_i(r)$ at every report profile $r\in T^n$ and for every $i\in I$. Thus, whenever $r(\sigma _{\tau }^{A}(r;i))>w(\sigma
_{\tau }^{A}(r;i))>t_{L}$ holds, the ``if'' part of Assumption \ref{NT} is satisfied. In other words, a firm will be able to detect a misreport by its matched partner. However, the possibility of such a detection does not immediately precipitate incentive compatibility. In order to make any progress, one needs the full force of the nontrivial updating assumption as the following example illustrates.

\begin{example}\label{ex-assumption1}
{\rm
Let $I=\{i_1,i_2\}$, $J=\{j_1,j_2\}$, and let $s_1=f(i_1)>f(i_2)=s_2$ be commonly known. Take $T=\{t_1,t_2,t_3,t_4\}$ with $t_1>t_2>t_3>t_4$. We will argue that the assortative matching rule $\sigma^A_{\tau}$ is not implementable by the AMIPA mechanism if the associated information structure violates nontivial updating (Assumption \ref{NT}). For this, we specify the information set of firm $i_1$ as follows: for any type assignment $w\in \mathcal{W}$, any matching $\mu\in \mathcal{M}$, any additional payoff relevant signal $\lambda_{i_1}\in 2^T$, and for $j\neq \mu(i_1)$, 
\begin{flalign}\label{eq1-assump-example}
\Pi_{i_1}(w,\mu,\lambda_{i_1}) &= \begin{cases}
                                       \{w'\in \clw \mid w'(\mu(i_1)=w(\mu(i_1))>w'(j)\} & \mbox{if } w(\mu(i_1))>t_4, \\
                                       \\
                                       \{w'\in \clw \mid w'(j)>w(\mu(i_1))\} & \mbox{if }w(\mu(i_1))=t_4.
                                     \end{cases}
\end{flalign}
\noindent
In other words, additional signals do not refine firm $i_1$'s information set. More importantly, in firm $i_1$'s assessment, its matched worker has the highest assigned type.  The only exception is when  the assigned type of the worker matched to the firm is the lowest possible type in $T$, i.e., when  $w(\mu(i_1))=t_4$.

Let us now consider the AMIPA mechanism $\langle \sigma^A_{\tau}, \alpha^I\rangle$. In what follows, we first show that the information set of firm $i_1$ defined above violates the nontrivial updating assumption with respect to the AMIPA mechanism. Next, as a consequence, we will show that in this case the assortative matching rule cannot be implemented by the AMIPA mechanism.\footnote{In fact, in this case, the assortative matching rule cannot be implemented by \textit{any} mechanism.} Recall that that for any report profile $r=(r(j_1),r(j_2))$, we have according to the AMIPA mechanism that $\alpha_i^I(r)=\{t\in T\mid t=r(j) \text{ for some } j\in J\}$ holds for each $i\in I$. Since the matching rule is assortative, as argued above, each firm can decipher the identity of the worker who has made a specific report. Consequently, and with a minor abuse of notation, we set $\alpha^I(r)=(\alpha_{i_{1}}^{I}(r),\alpha_{i_{2}}^{I}(r))=(r,r)$ for each $r\in T^n$. Now take the specific type assignment $w$ defined as follows: $w(j_1)=t_2>w(j_2)=t_3$. For the complete information assortative matching at $w$ we get $\{(i_1,j_1),(i_2,j_2)\}$.
Consider the report profile $\bar{r}$ where $j_2$ misreports via $t_1$: \[\bar{r}(j_1)=t_2 \mbox{ and } \bar{r}(j_2)=t_1.\] The matching at $\bar{r}$ will be $\sigma^A_{\tau}(\bar{r})=\{(i_1,j_2),(i_2,j_1)\}$. Since $\alpha^I(\bar{r})=(\bar{r},\bar{r})$, the profile of information sets is $\Pi(w,\sigma^A_{\tau}(\bar{r}),\bar{r})\big)$ and thus, the report of the matched partner $j_2$ of firm $i_1$ is revealed to $i_1$. Firm $i_1$ also observes that the true assigned type of its matched partner $j_2$ is $w(j_2)=t_3>t_4$. Therefore, at the report profile $\bar{r}$, firm $i_1$ is able to verify that its matched partner $j_2$ has misreported up, i.e., $\bar{r}(j_2)=t_1>w(j_2)=t_3>t_4$. So, all the conditions in the ``if'' part of Assumption \ref{NT} are satisfied.

 We will now demonstrate that $\big(w,\sigma^A_{\tau}(\bar{r}),\Pi(w,\sigma^A_{\tau}(\bar{r}),\bar{r})\big)$ is a stable matching state. This will then show that, at the type assignment $w$, worker $j_2$ can successfully manipulate the AMIPA mechanism via $\bar{r}(j_2)=t_1$. But note that this is immediate. The only blocking pair here must involve firm $i_1$. Now given the way firm $i_1$ forms its beliefs,  the mechanism's additional announcement does not lead to an update of its belief in \textit{any} way. In particular, the information set of firm $i_1$ at the type assignment $w$ given the matching $\sigma^A_{\tau}(\bar{r})$ and the public announcement $\alpha^I(\bar{r})=(\bar{r},\bar{r})$ is
 \begin{equation}\label{eq2-assump-example}
   \Pi_{i_1}(w,\sigma^A_{\tau}(\bar{r}),\bar{r})=\{w'\in \clw \mid w'(j_1)=t_4; \,w'(j_2)=t_3\}.
 \end{equation}
 That is, firm $i_1$ believes that the only type assignment given $w(j_2)=t_3$ is where $j_1$ has assigned type $t_4$. This is the case because the matched partner of $i_1$ has type $t_3$ which is not the lowest type and hence, only the first part of its information set in equation (\ref{eq1-assump-example}) applies. It is clear that equation \eqref{eq2-assump-example} is a violation of the nontrivial updating condition. But now note that, given the information set in equation \eqref{eq2-assump-example}, firm $i_1$ will never be a part of any blocking pair. So, the matching state $\big(w,\sigma^A_{\tau}(\bar{r}),\Pi(w,\sigma^A_{\tau}(\bar{r}),\bar{r})\big)$ is stable. Since $j_2$ can do this reasoning as well as firm $i_1$, worker $j_2$ can successfully manipulate the mechanism at $w$ via $\bar{r}(j_2)=t_1$. Therefore, the AMIPA mechanism is not incentive compatible.
}
\end{example}

\subsection{Non Coincidence of Reports}

We begin with the case in which report profiles are valid type assignments. That is, for any report profile $r\in T^{n}$, we have $r\in \mathcal{W}$. An immediate implication of this requirement is that no two workers can submit the same report. Admittedly, this is a restrictive assumption. However, it allows us to analyze a simplified environment first, before turning to the more general case in which coincident reports are permitted. To make
matters clear, let $\mathcal{R}$ denote the set of allowable report
profiles. In the case we consider in this subsection, $\mathcal{R}\subseteq
\mathcal{W}$. Since there is no coincidence of reports, the tie-breaking
rule $\tau $ does not play any role here and we therefore drop it from the
specification of the mechanism. Accordingly, the AMIPA mechanism takes the form $\langle \sigma ^{A},\alpha ^{I}\rangle $.

\begin{theorem} \label{Theorem1}
Let ${\cal R}\subseteq \mathcal{W}$. Then the AMIPA mechanism $\langle \sigma ^{A},\alpha ^{I}\rangle$ with associated information structure satisfying nontrivial updating is incentive compatible.
\end{theorem}

The proof of Theorem \ref{Theorem1} is relegated to  Appendix \ref{Proof-theorem1}. Below we provide an
example to demonstrate how the arguments in the proof work. Before we go
into the details of the example, we would like to again highlight two important aspects of the problem we consider. First, any worker $j$ of type $w(j)$, at the time of making a report $r(j)$, assumes that each of the other workers reports truthfully. Since $j$ is not assumed to know the realized type assignment, worker $j$ has to evaluate the consequences of her report for \textit{every type assignment} $w^{\prime }\in \mathcal{W}$ where $w^{\prime}(j)=w(j)$. Second, with the entire set of reported types being announced, each firm $i$ gets to observe the corresponding report profile as it is aware of the fact that the generated matching is the (positively)
assortative one with respect to the reported types. So, under the proposed
announcement, firm $i$ knows the type of its matched partner and the reports
of all the workers; especially, $i$ knows the reported types of all workers
who are willing to form a blocking pair with it. This rules out certain type
assignments as possible and leads to a refinement of firm $i$'s information
set. Since each worker can perform this reasoning as well as each firm, each
worker must take this reasoning into account when evaluating the possible
consequences of a (mis)-report. Exactly these two aspects, along with the
fact that the matching selected is assortative with respect to reports,
drive the result in Theorem \ref{Theorem1}.

\begin{example}\label{Exampleproof}
{\rm
Let $I=\left\{ i_{1},i_{2},i_{3}\right\} $ be
the set of firms and suppose their assigned types are commonly known with $%
s_{\ell }=f(i_{\ell })$ for each $\ell \in \left\{ 1,2,3\right\} $ and $%
s_{1}>s_{2}>s_{3}$. Let $J=\{j_{1},j_{2},j_{3}\}$ be the set of workers and $%
T$ the set of possible workers' types. Consider the AMIPA mechanism $%
\left\langle \sigma ^{A},\alpha ^{I}\right\rangle $ with the associated information structure satisfying nontrivial updating and let us provide an argument establishing incentive
compatibility of the mechanism. To that end, let $w=(t_{1},t_{2},t_{3})\in
\mathcal{W}$ be a specific realized type assignment with $t_{1}>t_{2}>t_{3}$%
. We make two claims. First, we show that a worker $j\in J$, whose realized
type is the second-ranked in $w$ (i.e., $w(j)=t_{2}$), cannot successfully
manipulate the mechanism (Claim 1). Given Claim 1, we then demonstrate that
a worker $j^{\prime }\in J$ whose realized type is the third-ranked in $w$
(that is, $w(j^{\prime })=t_{3}$) cannot successfully manipulate the
mechanism, either (Claim 2). Since we do not impose any restrictions on the
selection of $w\in \mathcal{W}$, the combination of these two claims
completes the proof of incentive compatibility of the AMIPA mechanism in
this particular example.\medskip \textrm{\ }\newline
\textrm{\noindent \textsc{Cl}}\textsc{aim 1}\textrm{\textsc{\textbf{\ }}}Let
$w=(t_{1},t_{2},t_{3})\in \mathcal{W}$ with $t_{1}>t_{2}>t_{3}$ be a
realized type assignment and suppose that, w.l.o.g., $w(j_{2})=t_{2}$. Then
it is impossible that worker $j_{2}$ successfully manipulates $\langle
\sigma ^{A},\alpha ^{I}\rangle $ at $w$.\medskip \newline
\textit{Proof of Claim }1. The assortative matching at the type assignment $%
w $ is $\sigma ^{A}(w)=\{(i_{1},j_{1}),(i_{2},j_{2}),(i_{3},j_{3}\}$.
Suppose now $j_{2}$ misreports at $w$ via some $\bar{t}_{2}>t_{1}$. The
resulting profile of reports would be $r=(r(j_{1}),r(j_{2}),r(j_{3}))=(t_{1},%
\bar{t}_{2},t_{3})$. At the report profile $r$, the generated assortative
matching would be%
\begin{equation*}
\sigma ^{A}(r)=\{(i_{1},j_{2}),(i_{2},j_{1}),(i_{3},j_{3})\}.
\end{equation*}%
The relevant matching state is then $\left( w,\sigma ^{A}(r),\Pi (w,\sigma
^{A}(r),\alpha ^{I}(r))\right) $. We will argue that this matching state is
not stable. For this, we show that the state is blocked by the pair $\left(
i_{1},j_{1}\right) $. Clearly, $j_{1}$ is willing to participate in this
blocking pair due to firms' types being publicly known and $i_{1}$ having
the highest possible type. As for the incentives of firm $i_{1}$ with $%
\alpha _{i_{1}}^{I}(r)=\left\{ t_{1},\bar{t}_{2},t_{3}\right\} $, we are
left with the following two mutually exclusive possibilities with respect to
its beliefs:

(1) there exists $w^{\prime }\in \Pi _{i_{1}}(w,\sigma ^{A}(r),\alpha
_{i_{1}}^{I}(r))$ with $\bar{t}_{2}>t_{2}>w^{\prime }(j_{1})$, or

(2) for all $w^{\prime }\in \Pi _{i_{1}}(w,\sigma ^{A}(r),\alpha
_{i_{1}}^{I}(r))$, $w^{\prime }(j_{1})>t_{2}$.

Notice that the existence of $w^{\prime }$ as described in the first case
would imply, by $t_{2}>t_{3}$ and each worker reporting upwards to the
mechanism (i.e., $t_{3}\geq w^{\prime }(j_{3})$), that worker $j_{2}$ is of
the highest realized type in the type assignment $w^{\prime }$. However, as
firm $i_{1}$ has detected $j_{2}$'s misreport by observing $\bar{t}%
_{2}>t_{2} $, we have a contradiction to the information structure satisfying nontrivial updating.
This excludes the first case,  while the
second case implies that $i_{1}$ and $j_{1}$ would from a blocking pair. We
conclude that the matching state $\left( w,\sigma ^{A}(r),\Pi (w,\sigma
^{A}(r),\alpha ^{I}(r))\right) $ is not stable and thus, worker $j_{2}$
cannot successfully manipulate $\langle \sigma ^{A},\alpha ^{I}\rangle $ at $%
w$.\medskip\ \newline
\noindent \textsc{Claim 2}\textrm{\textsc{\textbf{\ }}}Let $%
w=(t_{1},t_{2},t_{3})\in \mathcal{W}$ with $t_{1}>t_{2}>t_{3}$ be a realized
type assignment and suppose that, w.l.o.g., $w(j_{3})=t_{3}$. Then it is
impossible that worker $j_{3}$ successfully manipulates $\left\langle \sigma
^{A},\alpha ^{I}\right\rangle $ at $w$.\medskip\ \newline
\textit{Proof of Claim }2. A manipulation by worker $j_{3}$ at the type
assignment $w$ would imply the existence of a type $\bar{t}_{3}$ such that
either (1) $t_{1}>\bar{t}_{3}>t_{2}>t_{3}$ or (2) $\bar{t}%
_{3}>t_{1}>t_{2}>t_{3}$ is the case. We call the first possibility \textit{%
swap-1 manipulation} and the second possibility \textit{swap-2 manipulation}%
.\medskip \newline
\textsc{Swap-1 manipulation: }Note that here the report profile is $%
r=(t_{1},t_{2},\bar{t_{3}})$ with $t_{1}>\bar{t}_{3}>t_{2}>t_{3}$. The
generated assortative matching would be $\sigma
^{A}(r)=\{(i_{1},j_{1}),(i_{2},j_{3}),(i_{3},j_{2})\}$. We will argue that
the relevant matching state $\left( w,\sigma ^{A}(r),\Pi (w,\sigma
^{A}(r),\alpha ^{I}(r))\right) $ is not stable. For this, we show that the
state is blocked by the pair $\left( i_{2},j_{2}\right) $. Clearly, $j_{2}$
is willing to participate in this blocking pair due to firms' types being
publicly known and $i_{2}$ being of a higher type in comparison to $i_{3}$.
As for the incentives of firm $i_{2}$, notice that this firm in particular
observes the reports of worker $j_{3}$ ($\bar{t}_{3}$) and of worker $j_{2}$
($t_{2}$) due to $\alpha _{i_{2}}^{I}(r)=\left\{ t_{1},\bar{t}%
_{3},t_{2}\right\} $, as well as the true type $w(j_{3})=t_{3}$ of its
matched worker $j_{3}$. Hence, there are four mutually exclusive cases for
any $w^{\prime }\in \Pi _{i_{2}}(w,\sigma ^{A}(r),\alpha _{i_{2}}^{I}(r))$:

(a) $t_{3}>w^{\prime }(j_{1}),w^{\prime }(j_{2})$,

(b) $w^{\prime }(j_{1})>t_{3}>w^{\prime }(j_{2})$,

(c) $w^{\prime }(j_{1}),w^{\prime }(j_{2})>t_{3}$, or

(d) $w^{\prime }(j_{2})>t_{3}>w^{\prime }(j_{1})$.

Now case (a) is ruled out by the nontrivial updating assumption. Note next that there
can be two possibilities for any $w^{\prime }$ satisfying the condition in
case (b). The first possibility of $w^{\prime }(j_{1})>\bar{t}%
_{3}>t_{3}>w^{\prime }(j_{2})$ is again ruled out by the nontrivial updating assumption (since there is no worker whose $w^{\prime }$-assigned type is between $%
t_{3} $ and $\bar{t}_{3}$). The other possibility is $\bar{t}_{3}>w^{\prime
}(j_{1})>t_{3}>w^{\prime }(j_{2})$. But then worker $j_{3}$ of type $t_{3}$
is misreporting at a type assignment where her type is second ranked, a
situation already excluded by Claim 1. The consequence of cases (a) and (b) being
ruled out is that, for any $w^{\prime }\in \Pi _{i_{2}}(w,\sigma
^{A}(r),\alpha _{i_{2}}^{I}(r))$, either case (c) or case (d) holds. But in
any of these cases, $w^{\prime }(j_{2})>t_{3}$. Thus, firm $i_{2}$ and
worker $j_{2}$ would form a blocking pair. We conclude that the matching
state $\left( w,\sigma ^{A}(r),\Pi (w,\sigma ^{A}(r),\alpha ^{I}(r))\right) $
is not stable and thus, worker $j_{3}$ cannot successfully swap-1 manipulate
$\langle \sigma ^{A},\alpha ^{I}\rangle $ at $w$.\medskip \newline
\textsc{Swap-2 Manipulation:} The report
profile here is $r=(t_{1},t_{2},\bar{t_{3}})$ with $\bar{t}%
_{3}>t_{1}>t_{2}>t_{3}$. The resulting assortative matching would be $\sigma
^{A}(r)=\{(i_{1},j_{3}),(i_{2},j_{1}),(i_{3},j_{2})\}$. We will argue that the resulting matching state, $\big(w,\sigma^A(r),\Pi(w,\sigma^A(r),\alpha^I(r))\big)$ is not stable.

In this case, any blocking pair, must involve firm $i_1$. Note that under the AMIPA mechanism firm  $i_1$ observes the report profile  $r$ and $w(j_3)=t_3$. Now an analogous argument\footnote{%
The complete argument is presented in the proof of Theorem \ref{Theorem1}.}
 as in the case of swap-1 manipulation leads to the conclusion that $(i_1,j_2)$ would form a blocking pair for the matching state  $\big(w,\sigma^A(r),\Pi(w,\sigma^A(r),\alpha^I(r))\big)$.
}
\end{example}

\subsection{Coincidence of Reports}

We now allow for the possibility that different workers may submit the same report to the mechanism. Consequently, the set of admissible report profiles becomes $\mathcal{R}\equiv T^n$. In this case, the tie-breaking rule $\tau$ is clearly required as part of the mechanism.

\begin{theorem}\label{Theorem2}
Let $\mathcal{R}=T^n$. Then the AMIPA mechanism $\langle \sigma^A_{\tau}, \alpha^I \rangle$ with associated information structure satisfying nontrivial updating is incentive compatible.
\end{theorem}

The proof of Theorem \ref{Theorem2} is relegated to Appendix \ref{Proof-theorem2}. It is worth noting that the nontrivial updating assumption (Assumption \ref{NT}), which is crucial for both Theorems \ref{Theorem1} and \ref{Theorem2}, imposes some “alignment of beliefs” between a firm and the worker to whom it is matched. For example, consider a worker $j$ with true type $w(j)$ who contemplates a misreport $t$, assuming all other workers report truthfully. From $j$’s perspective, such a misreport would result in a matching where firm $i$ observes both $j$’s true type $w(j)$ and the misreported type $t$. To assess the consequences of this misreport, worker $j$ must form beliefs about the beliefs of firm $i$ when encountering the misreport. Since the information structure is commonly known, $j$ understands how $i$ forms these beliefs. Nontriviality thus requires that firm $i$ updates its beliefs to the subset of type assignments for which the misreport could be “beneficial” to worker $j$. This alignment of beliefs between firms and their matched workers underpins the incentive compatibility of the AMIPA mechanism. 

We identify two potential approaches to dispense with the nontrivial updating condition. The first approach assumes that the type assignment is commonly known among workers. If each worker, when reporting, assumes that all other workers report truthfully, then common knowledge of the type assignment ensures that at any report profile there can be at most one misreport. This directly generates the alignment of beliefs described above. The second approach assumes that each firm believes that every worker, except for the one to whom it is matched, reports truthfully. Although this setup allows for the possibility of multiple misreports at a report profile, the assumption effectively precludes such scenarios. 

In other words, replacing the nontrivial updating condition with either of the assumptions outlined above would require imposing another strong requirement on the information structure. While Assumption \ref{NT} is admittedly strong, we consider it reasonable and proceed with our analysis under this condition.

\section{Private Announcements}

In this section, we explore the possibility of relaxing the requirement that the entire set of reported types be publicly announced. As noted above, this is equivalent to publicly revealing the entire report profile and has two key consequences. First, any firm can detect a potential misreport by its matched worker. Second, because the firm observes the true type of its matched worker, it can compare this observed type with the reports of other workers. Consider a firm $i$ matched to a worker $j$ at the report profile $r$, and suppose that $j$ has misreported ``up'', i.e., $r(j)>w(j)$. Any worker $j'$ who reports $r(j')>r(j)$ will be matched with a firm of higher type than $i$ and therefore cannot form a blocking pair with $i$. Hence, for firm $i$, the only relevant portion of the publicly announced report profile is the subset of reports that  are ``lower'' than $r(j)$.

Motivated by this observation, we consider mechanisms in which each firm observes only the subset of reports that are lower than the report of its matched worker, and possibly the report of that worker itself. The reasoning above suggests that this “limited information” may be sufficient to ensure incentive compatibility. Because these announcements are firm-specific and vary across firms, we refer to them as \textit{private announcements}.

\subsection{Lower Contour Sets Announcements}

We start by defining the mechanism $\langle \sigma _{\tau }^{A},\alpha
^{L}\rangle $ where $\sigma _{\tau }^{A}$ is the assortative matching rule
and $\alpha ^{L}=\left( \alpha _{i}^{L}\right) _{i\in I}$ is the collection
of announcement rules such that, for every report profile $r\in T^{n}$,
\begin{equation}\label{PAeq1}
\alpha _{i}^{L}(r)=\begin{cases}
                     \left\{ t\in T\mid t=r(j)\leq r(\sigma _{\tau }^{A}(r;i))\text{ for some }%
j\in J\right\} & \mbox{if } i\in I\setminus \left\{ i_{n}\right\} , \\
                     \emptyset & \mbox{if }i=i_{n}.
                   \end{cases}
\end{equation}

In words, each firm $i\neq i_n$ gets to know the assortative matching for each profile of reports, together with all reports that are in $i$'s \textit{lower contour set} of reported types. Given that the matching rule $\sigma^A_{\tau}$ is assortative, this is equivalent to saying that each firm $i\neq i_n$ gets to observe the sub-profile of reports that are \textit{weakly} lower than the report of its matched partner. Any kind of announcement by the mechanism for the lowest ranked firm $i_n$ is irrelevant, since $i_n$ will never be a part of a blocking pair. This is the reason for setting $\alpha _{i_{n}}^{L}(r)=\emptyset $ for all $%
r\in T^{n}$ in the definition of $\langle \sigma _{\tau }^{A},\alpha
^{L}\rangle $.

Observe that the above mechanism enables each firm $i\neq i_n$ to (i) detect a possible misreport by its matched worker, and (ii) to compare the \textit{position} of the true assigned type of the matched worker in the corresponding sub-profile of types in the lower contour set of reported types. Building on this insight, we present the following result, which is an immediate consequence of Theorem \ref{Theorem1} and Theorem \ref{Theorem2}.

\begin{prop}\label{PropLC}
The assortative matching mechanism $\langle \sigma^A_{\tau},\alpha^L \rangle$ with associated information structure satisfying nontrivial updating is incentive compatible.
\end{prop}

Consider a type assignment $w$ and a profile of reports $r$. When the mechanism reveals to firm $i\neq i_n$ the entire lower contour set of reports (including the report of its matched partner), firm $i$ is able to settle two issues. First, firm $i$ can detect a potential misreport by its matched partner. Second, and more importantly, firm $i$ is able to compare the realized type of its matched partner $w(\sigma^A_{\tau}(r))$ with the reports of the remaining workers in the lower contour set. Therefore, firm $i$ can infer that any worker $j'\neq \sigma^A_{\tau}(r;i)$ whose report is lower than the type of its matched partner $ \sigma^A_{\tau}(r;i)$ indeed has a true assigned type lower than $w( \sigma^A_{\tau}(r;i))$. As a consequence, firm $i$ would rule out any such worker $j'$ as a potential blocking partner. As shown in Theorem \ref{Theorem1} and Theorem \ref{Theorem2}, the idea described above leaves  firm $i$ only to consider those workers for potential blocking partners, whose reports lie in between the report $r(\sigma^A_{\tau}(r;i))$ and the true type $w(\sigma^A_{\tau}(r;i))$ of $i$'s matched partner. 

\subsection{Mere Verification of Misreports}

To highlight the role of reports in the strict lower contour sets, we examine the assortative matching mechanism in which each firm receives a private announcement revealing only the type of its matched worker. This mechanism is denoted by $\langle \sigma^A_{\tau}, \alpha^{\triangle}\rangle $ where $\alpha^{\triangle}= \big(\alpha^{\triangle}_i\big)_{i\in I}$ is the collection of announcement rules such that for each profile of reports $r\in T^n$,
\[
\alpha^{\triangle}_i=\begin{cases}
                       r(\sigma^A_{\tau}(r;i)) & \mbox{if } i\neq i_n, \\
                       \emptyset & \mbox{otherwise}.
                     \end{cases}
\]

Clearly, each firm can verify whether its matched worker has reported truthfully. However, unlike under the mechanism $\langle \sigma _{\tau }^{A},\alpha^{L}\rangle $, the firm cannot determine \textit{the position} of that worker's true assigned type relative to the reports in its lower contour set. Our next result demonstrates, via an example, that the matching mechanism $\langle \sigma^A_{\tau}, \alpha^{\triangle}\rangle $ is not incentive compatible when there are at least three agents on each side of the market. We emphasize that this negative result does not arise from a violation of Assumption \ref{NT}; indeed, we assume throughout that the information structure satisfies nontrivial updating.

 \begin{prop}\label{Proposition3}
Let $|I|=|J|\geq 3$. Then the assortative matching mechanism $\langle \sigma^A_{\tau}, \alpha^{\triangle}\rangle$ with associated information structure satisfying nontrivial updating is not incentive compatible.
\end{prop}

\begin{proof}
Let $I=\{i_{1},i_{2},i_{3}\}$ and $J=\{j_{1},j_{2},j_{3}\}$. The ordering of the firms according to their types
is commonly known and set to be $f(i_{1})>f(i_{2})>f(i_{3})$. Let $T=\{t_{1},\cdots ,t_{7}\}$ be the set of possible worker types with $t_{1}>t_{2}>\cdots >t_{6}>t_{7}$. Consider now the type assignment $w$ defined as follows: $w(j_1)=t_2$, $w(j_2)=t_4$ and $w(j_3)=t_5$. If every worker reports truthfully, the matching would be $\sigma^A_{\tau}(w)=\{(i_{1},j_{1}),(i_{2},j_{2}),(i_{3},j_{3})\}$. Now suppose that worker $j_2$ contemplates misreporting at type assigment $w$ via $t_1>t_4=w(j_2)$.  From the perspective of worker $j_2$, and given the assumption that every other worker is reporting truthfully, the resulting profile of reports $r$ would be $r(j_1)=t_2$, $r(j_2)=t_1$ and $r(j_3)=t_5$. The resulting matching would then be $\sigma^A_{\tau}(r)=\{(i_{1},j_{2}),\,(i_{2},j_{1}),(i_{3},j_{3})\}$. Each firm $i\neq i_3$ would get to know the report of its matched worker, i.e., $\alpha _{i_{1}}^{\triangle }(r)=\left\{ t_{1}\right\}$, $\alpha_{i_{2}}^{\triangle }(r)=\left\{ t_{2}\right\} $, $\alpha_{i_{3}}^{\triangle }(r)=\emptyset $. Moreover, given that the matching rule is assortative with respect to reports, each firm $i\neq i_3$ would get to infer that $r(j_2)>r(j_1)>r(j_3)$ holds. It is clear that firm $i_1$ observes that it's matched partner $j_2$ has misreported: $w(j_2)=t_4< t_1=r(j_2)$. The proof of the proposition will be complete if we can establish that the matching state  $\big(w,\sigma _{\tau}^{A}(r),\Pi \left( w,\sigma _{\tau }^{A}(r),\alpha ^{\triangle }(r)\right)\big)$ is stable. 

Note that neither firm $i_3$, nor firm $i_2$ can be a part of any blocking pair.  Firm $i_3$ has the lowest possible type. Consequently, firm $i_3$ cannot form any successful block with any worker $j$. Note that firm $i_2$ knows that worker  $j_2$ is matched to the firm with the highest possible type (firm $i_1$). Hence, the only possible blocking pair that firm $i_2$ can be a part of is $(i_2,j_3)$. However, given that $t_2=w(j_1)=r(j_1)>t_3=r(j_3)\geq w(j_3)$, firm $i_2$ will not be willing to form a blocking pair with worker $j_3$. So, the only possible blocking pairs to consider  involve firm $i_1$. 

Now each of the workers $j_1$ and $j_3$ would like to form a blocking pair with firm $i_1$ (since firm $i_1$ is the firm with the highest possible type). On the other hand, firm $i_1$ would form a blocking pair with any worker $j\in \{j_1,j_3\}$ if $i_1$ is \textit{certain} that the assigned type of $j$ is higher than  $w(j_2)=t_4$. With the information structure satisfying nontrivial updating, and given the misreport by worker $j_2$, firm $i_1$ would believe that in all possible type assignments either $j_1$ or $j_3$ (or both) have assigned true type higher than $t_4$. Moreover, firm $i_1$ would only block the matching $\sigma^A_{\tau}(r)$ with some worker $j\in \{j_1,j_3\}$ if for all type assignments $\tilde{w}\in \Pi_{i_1}(w,\sigma^A_{\tau}(r),\alpha^{\triangle}(r))$ it is the case that $\tilde{w}(j)> \tilde{w}(j_2)=w(j_2)=t_4$. In what follows, we demonstrate that $i_1$ will not be able to identify any such worker $j\in \{j_1,j_3\}$. 

For this, recall that firm $i_{1}$ only observes $j_{2}$'s realized type $%
t_{4}$ as well as $j_{2}$'s (mis)report $t_{1}$ to the mechanism; moreover, $%
i_{1}$ can derive the ordering of the reported types. However, $i_{1}$ will
not be able to distinguish between the type assignment $w$ and another type
assignment $w^{\prime }$, provided that workers' reports at $w^{\prime }$
generate the same matching $\sigma _{\tau }^{A}(r)$ as above with $i_{1}$
being matched to $j_{2}$ of true type $t_{4}$ and (mis)report $t_{1}$.

To be precise, consider the following type assignment $w'$ where $w'(j_1)=t_5$, $w'(j_2)=t_4$ and $w'(j_3)=t_3$. Observe that in both type assignments $w$ and $w'$, the assigned type of worker $j_2$ is the same, $w(j_2)=w'(j_2)=t_4$. Further, firm $i_1$ cannot distinguish between the following two situations: (i) worker $j_2$ misreporting at $w$ via $t_1$, and (ii) worker $j_2$ misreporting at $w$ via $t_1$ \textit{and} worker $j_1$ misreporting at $w'$ via $t_2$. The reason here is that the resulting matching in both situations is $\{(i_{1},j_{2}),\,(i_{2},j_{1}),(i_{3},j_{3})\}$ with the mechanism conveying $j_2$'s report $t_1$ to firm $i_1$, $j_1$'s report $t_2$ to firm $i_2$, and providing no additional information to firm $i_3$. That is, 
\[
w,w^{\prime }\in \Pi _{i_{1}}(w,\sigma _{\tau }^{A}(r),\alpha
_{i_{1}}^{\triangle }(r)).
\]

Moreover, in the type assignment $w$, we have  $w(j_{1})=t_{2}>t_{4}>t_{5}=w(j_{3})$ while in  the type assignment $w'$, we have $w^{\prime}(j_{3})=t_{3}>t_{4}>t_{5}=w^{\prime }(j_{1})$. Consequently, given nontrivial updating, even though firm $i_{1}$ knows that at the realized type assignment $w$ either worker $j_{1}$ or worker $j_{3}$ has assigned type
higher than $j_{2}$'s realized type $t_{4}$, it cannot identify the worker who has the higher type. That is, firm $i_{1}$ will not be a part of any blocking pair for the matching state $(w,\sigma _{\tau }^{A}(r),\Pi \left(w,\sigma _{\tau }^{A}(r),\alpha ^{\triangle}(r)\right) )$.   Since $j_{2}$ can perform the above analysis as well as firm $i_{1}$, worker $j_{2}$ can
successfully manipulate the mechanism $\langle \sigma _{\tau }^{A},\alpha
^{\triangle }\rangle $ at the type assignment $w$ via report $r(j_{2})=t_{1}$. Hence, this matching mechanism is not incentive compatible.
\end{proof}
 
 \section{Concluding Remarks}

In this paper, we study a two-sided matching model with one-sided incomplete information. Under a permissive notion of stability, almost any matching can be stable; by contrast, in the complete information case, the positively assortative matching is the unique stable allocation. This observation motivates our mechanism design approach, focusing on the existence of incentive compatible mechanisms that implement the positively assortative matching. It is immediate that, without conveying some information about the possible type distribution, such an allocation cannot be implemented. In our framework, information is conveyed through announcements to firms of their corresponding lower contour sets of reported types. We identify a condition on the information structure - nontrivial updating - under which the assortative matching rule becomes implementable.

A natural question is whether the positively assortative matching allocation can be implemented when both sides of the market have incomplete information. The expected answer is negative, since even in the standard Gale-Shapley model there is a fundamental tension between incentive compatibility on both sides and stability (\citet{Roth1}). We conclude the paper by examining this issue.

In order to formulate the problem appropriately, we assume that firms' types
are also private information and denote by $\mathcal{F}$ the set of all
possible assignments $f:I\rightarrow S$ of firms' types with $\left\vert
S\right\vert >n$. A general type assignment is now a pair $(f,w)$ where $%
f\in \mathcal{F}$ and $w\in \mathcal{W}$. A mechanism then would ask for
reports from both firms and workers. Thus, a generic report profile is $%
r\equiv \left( (r(i))_{i\in I},(r(j))_{j\in J}\right) \in S^{n}\times T^{n}$%
. We consider the case where the matching mechanism generates once again the
(positively) assortative matching with respect to the reported types and
publicly announces the set of all reported types. The corresponding
information structure is denoted by $\Pi ()=\left( (\Pi _{i}())_{i\in
I},(\Pi _{j}())_{j\in J}\right) $. Moreover, we extend the tie-breaking rule $\tau $ to
take care of coincidence of reports on either side of the market and,
slightly abusing notation, denote the corresponding AMIPA mechanism again by
$\langle \sigma _{\tau }^{A},\alpha ^{I}\rangle $. Our final example
establishes that this mechanism fails incentive compatibility in the
considered extended setting.\medskip \newline
\textsc{Example \textbf{6 }}Let $I=\{i_{1},i_{2}\}$ and $J=\{j_{1},j_{2}\}$.
The sets of workers' and firms' types are $T=\{t_{1},t_{2},t_{3},t_{4}\}$
and $S=\{s_{1},s_{2},s_{3},s_{4}\}$, respectively. The types in these sets
are ordered as follows: $t_{1}>t_{2}>t_{3}>t_{4}$ and $%
s_{1}>s_{2}>s_{3}>s_{4}$.

Consider the general type assignment $(f,w)$ where $f(i_{1})=s_{2}$, $%
f\left( i_{2}\right) =s_{3}$, $w(j_{1})=t_{2}$, and $w(j_{2})=t_{3}$. Let $%
r^{0}$ be the truthful report profile at $(f,w)$, that is, $r^{0}(i)=f(i)$
for each $i\in I$ and $r^{0}(j)=w(j)$ for each $j\in J$. With truthful
reports, the assortative matching would be $\sigma _{\tau
}^{A}(r^{0})=\{(i_{1},j_{1}),(i_{2},j_{2})\}$. Now suppose worker $j_{2}$
misreports at $(f,w)$ via $t_{1}$. Then the report profile is $r$, where $%
r(j_{2})=t_{1}$ and for all $k\in \left( I\cup J\right) \setminus \left\{
j_{2}\right\} $, $r(k)=r^{0}(k)$. The matching at the report profile $r$
would be $\sigma _{\tau }^{A}(r)=\{(i_{1},j_{2}),(i_{2},j_{1})\}$. Moreover,
every agent knows the report profile and firm $i_{1}$ knows that its matched
partner $j_{2}$ has misreported. If we additionally impose nontrivial updating on the information structure, firm $i_{1}$ would even know that worker $j_{1}$ is of a type
higher than $w(j_{2})=t_{3}$. Despite this, we will argue below that the
matching state $\left( (f,w),\sigma _{\tau }^{A}(r),\Pi ((f,w),\sigma _{\tau
}^{A}(r),\alpha ^{I}(r))\right) $ is stable.

Notice first that $\left( i_{2},j_{2}\right) $ cannot be a blocking pair,
the reason being that $j_{2}$ knows the true type $f(i_{1})=s_{2}$ of its
matching partner as well as the reported type $r(i_{2})=s_{3}$ of firm $%
i_{2} $. Since the matching rule is assortative with respect to the
reported types and everyone is aware of this fact, the true type of firm $%
i_{2}$ is weakly lower than $s_{3}$ and thus, worker $j_{2}$ will not be
willing to form a blocking pair with $i_{2}$.

Consider now the possibility of $(i_{1},j_{1})$ being a blocking pair and
notice that firm $i_{1}$ would like to form a block with $j_{1}$ for reasons
based on nontrivial updating and elaborated above. Observe that, at the report profile $r$, worker $j_{1}$ knows that her matched firm $%
i_{2}$ has reported truthfully. More importantly, $j_{1}$ has no way to
verify whether $i_{1}$ has reported truthfully or has misreported. Nontrivial updating or an analogous assumption for workers has little impact in this
case. Consider therefore the firms' type assignment $f^{\prime }=\left(
f^{\prime }(i_{1}),f^{\prime }(i_{2})\right) =\left( s_{4},s_{3}\right) $
and note that, as argued above, $(f^{\prime },w)\in \Pi
_{j_{1}}((f,w),\sigma _{\tau }^{A}(r),\alpha _{j_{1}}^{I}(r))$. Hence,
worker $j_{1}$ will not be a part of any blocking pair. This leads to the
conclusion that the matching state $\left( (f,w),\sigma ^{A}(r),\Pi
((f,w),\sigma _{\tau }^{A}(r),\alpha ^{I}(r))\right) $ is stable. Since $%
j_{2}$\ can perform the above analysis as well as $j_{1}$, worker $j_{2}$
can successfully manipulate the mechanism at $(f,w)$ via $t_{1}$.

\newpage
\begin{appendices}
\section*{Appendix}
\section{Proof of Theorem \ref{Theorem1}}\label{Proof-theorem1}

We start with the following additional notions and definitions. Consider a type assignment $w\in \mathcal{W}$ and fix a worker $j\in J$. Let $u(w;j)=|\{k\in J\mid w(k)>w(j)\}|$ be the number
of workers whose realized types are higher than the realized type of worker $%
j$. The \textit{position} of worker $j$ in the type assignment $w$ is $%
p_{j}(w)=u(w;j)+1$. For any $q\in\left\{1,\cdots,n\right\}$, define the set $\mathcal{P}%
(w(j);q)=\{w^{\prime }\in \mathcal{W}\mid w^{\prime }(j)=w(j)\text{ and }%
p_{j}(w)=q\}$. Thus, $\mathcal{P}(w(j);q)$ is the set of type assignments
where worker $j$ of realized type $w(j)$ is at the $q$-th position in $w$.
Note that $\{\mathcal{P}(w(j);q)\}_{q=1}^{n}$ partitions the space of type
assignments. Also note that $w\in \mathcal{P}(w(j);1)$ implies that $j$ has
no reason to misreport at $w$. That is, worker $j$ will consider misreporting only if
she believes that the true type assignment $w$ belongs to $\mathcal{P}%
(w(j);q)$ for some $q>1$. In our proof argument, we will first
demonstrate that, at any $w\in \mathcal{P}(w(j);2)$, worker $j$ cannot
successfully manipulate. We then complete our argument by induction on $q$.
Notice therefore that, in our formulation of the problem, a successful
manipulation by any worker $j$ only pertains to improving the worker's
position in the report profile. Such a manipulation is possible only because
worker $j$ at the time of reporting assumes every other worker is reporting
truthfully, and this is commonly known. In order to make matters clear, we
begin with the following definitions.

\begin{defn}
Let $k\in \left\{ 1,2,\cdots,n-1\right\} $. We say that worker $j\in J$ can \textbf{potentially swap-$k$ manipulate} an assortative matching mechanism $\langle \sigma ^{A},\alpha \rangle $ at a type assignment $w\in \mathcal{W}$  via  $t\in T$ if $u((w_{-j},t);j)=u(w;j)-k$. We say that a potential swap-$k$ manipulation is \textbf{successful}, if the matching state $\left( w,\sigma^{A}(w_{-j},t),\Pi (w,\sigma ^{A}(w_{-j},t),\alpha (w_{-j},t))\right) $  is stable.
\end{defn}

\begin{defn}\label{defn-nonmani}
   Fix $\left( p,k\right) $ with $
p\in \{ 2,\cdots ,n\} $ and $k\in \{ 1,\cdots
,p-1\} $. We say that an assortative matching mechanism $\langle \sigma ^{A},\alpha \rangle $ is

\begin{itemize}
  \item [-]  $( p,k) $\textbf{-manipulable} at a realized
type assignment $w \in \mathcal{W} $ if some worker $j\in J$ with
$p_{j}(w )=p$ can successfully swap-$k$ manipulate
at $w$.
  \item [-] $ (p,k)$\textbf{-non-manipulable} if for all
realized type assignments $w \in \mathcal{W} $ and for all $k\leq
p-1 $, no worker $j\in J$  with  $p_{j}(w )=p$
can successfully swap-$k$ manipulate at $w $.
\end{itemize}
\end{defn}

In view of Definition \ref{defn-nonmani}, we can say that an assortative matching mechanism $\langle \sigma ^{A},\alpha \rangle $ is incentive compatible, if it is $(p,k)$-non-manipulable for every $p\in \{2,\cdots,n\}$. We plan to demonstrate that the AMIPA mechanism $\langle \sigma^A,\alpha^I\rangle$ is incentive compatible by showing that $\langle \sigma^A,\alpha^I\rangle$ is $(p,k)$-non-manipulable for every $p$.

\medskip 
\noindent
{\sc Proof of Theorem \ref{Theorem1}:} 
We proceed by induction. More precisely, we fix the AMIPA mechanism $\langle \sigma ^{A},\alpha
^{I}\rangle $ and first show that it is $\left( 2,1\right) $-non-manipulable
(Statement A). In the Induction Hypothesis we assume that for all $%
p<p^{\prime }$ the mechanism is $\left( p,k\right) $-non-manipulable and
then show (Statement B) that the mechanism is $\left( p^{\prime },k\right) $%
-non-manipulable. 

\vskip 0.2cm 
\noindent
\textsc{Statement A}\textit{\hspace{0.2 cm}The AMIPA mechanism $\langle\sigma^A,\alpha^I\rangle$ is $(2,1)$-non-manipulable.}

\medskip 
\noindent
\begin{proof}
Note that by the definition of the AMIPA mechanism firms
know the entire report profile. In other words, at report profile $r\in
T^{n} $, each firm $i\in I$ gets to observe the report $r(j)$ by each worker 
$j\in J$. Since in the environment we consider in Theorem 1, $\mathcal{R}%
\subseteq \mathcal{W}$, each firm $i\in I$ is thus observing the \textit{%
identity} of the worker $j$ whose reported type is $r(j)$. Additionally,
each firm $i\in I $ gets also to observe the true type of its matched
partner $w(\sigma ^{A}(r;i))$ and is aware of the fact that $r(j)\geq w(j)$
holds for each worker $j\in J$. Notice that the last inequality immediately
follows from the assortativeness of the commonly known matching rule.

We now proceed with the proof of Statement A. Let $w\in \mathcal{W}$ be a
type assignment and let $j$ be the second ranked worker in $w$, i.e., $%
p_{j}(w)=2$ and thus, $w\in \mathcal{P}(w(j);2)$. Note that a $(2,1)$%
-manipulation is meaningful only in this scenario. Without loss of
generality, let $j=j_{2}$ and $w=(t_{1},t_{2},\cdots ,t_{n})$ with $t_{\ell
}=w(j_{\ell })$ be such that $t_{1}>t_{2}>\cdots >t_{n}$. Suppose that
worker $j_{2}$ is contemplating a manipulation at $w$ via $t>t_{1}$.%
\footnote{%
If there does not exist such a $t>t_{1}=w(j_{1})$, then Statement A is
trivially satisfied.} We will show that such a misreport will not be
successful.

Given our notion of incentive compatibility, worker $j_{2}$ reports $%
r(j_{2})=t$ under the assumption that every other worker is reporting
truthfully. So, from the perspective of $j_{2}$, the report profile is $%
r=(w_{-j_{2}},t)=(t_{-2},t)$. For the matching $\sigma ^{A}(t_{-2},t)$, we
get $(i_{1},j_{2}),(i_{2},j_{1})$ and $(i_{\ell },j_{\ell })$ for all $\ell
\in \{3,\cdots ,n\}$. Given worker $j_{2}$'s report $r(j_{2})=t$ and the
matching, firm $i_{1}$ would know the entire report profile $r$ and worker $%
j_{2}$'s true type $w(j_{2})=t_{2}$. Thus, firm $i_{1}$ observes the profile
of reports $r=(t_{-2},t)$ and that it's matched partner $j_{2}$ has
misreported via $t$ when $j_{2}$'s true type is $w(j_{2})$. Firm $i_{1}$
does not know the true type assignment but knows that each worker $j$, when
making her report, assumes every other worker is reporting truthfully. In
particular, $j_{2}$ has misreported via $t$ assuming $r_{-j_{2}}=t_{-2}$ as
the true type assignment of all other workers.

Let us now consider the information set $\Pi _{i_{1}}(w,\sigma
^{A}(r),\alpha _{i_{1}}^{I}(r))$ of firm $i_{1}$ and any type assignment $%
w^{\prime }$ belonging to that set. Recall that $w^{\prime
}(j_{2})=t_{2}=w(j_{2})$ and as noted above, for all $\ell \in \{3,\cdots
,n\}$,%
\begin{equation*}
t_{2}=w^{\prime }(j_{2})>t_{\ell }\geq w^{\prime }(j_{\ell }).
\end{equation*}%
In addition, $w^{\prime }$ is a possible realized type assignment from the
point of view of $i_{1}$. Firm $i_{1}$ can therefore conclude that%
\begin{equation*}
w(j_{2})>w(j_{\ell })\geq w^{\prime }(j_{\ell })\text{ for all }\ell \in
\{3,\cdots ,n\}.
\end{equation*}%
It only remains to consider $w^{\prime }(j_{1})$. We make the following
claim.\medskip \newline
\textsc{Claim 1 }$w^{\prime }(j_{1})>t_{2}=w^{\prime }(j_{2})$.\medskip 
\newline
\noindent \textit{Proof of Claim }1. Suppose, contrary to the claim, that $%
t_{2}>w^{\prime }(j_{1})$ holds. Then we have $t>t_{2}=w^{\prime
}(j_{2})>w^{\prime }(j_{1})$. This implies that worker $j_{2}$ has the
highest assigned type in $w^{\prime }$. At such a profile $w^{\prime }$, it
is not worthwhile for $j_{2}$ to misreport via $t$. Given that $j_{2}$ has
indeed misreported, $w^{\prime }\notin \mathcal{W}%
(w(j_{2}),r(j_{2})=t)$ must be the case. Nontrivial updating then requires
$\Pi _{i_{1}}(w,\sigma ^{A}(r),\alpha _{i_{1}}^{I}(r))\subseteq \mathcal{W}%
(w(j_{2}),r(j_{2})=t)$ and thus, $w^{\prime }\notin \Pi
_{i_{1}}(w,\sigma ^{A}(r),\alpha _{i_{1}}^{I}(r))$ holds. This contradicts our starting point $w^{\prime }\in \Pi _{i_{1}}(w,\sigma ^{A}(r),\alpha
_{i_{1}}^{I}(r))$ and completes the proof of the claim. 

\medskip 
\noindent 
In view of the above claim, we have that for all $w^{\prime }\in
\Pi _{i_{1}}(w,\sigma ^{A}(r),\alpha _{i_{1}}^{I}(r))$, it must be the case
that $w^{\prime }(j_{1})>t_{2}=w^{\prime }(j_{2})=w(j_{2})$. Recalling that $j_{1}$ is matched to $i_{2}$ and $s_{1}>s_{2}$, we conclude that $(i_{1},j_{1})$ is a blocking pair. As $j_{2}$ can perform the above analysis
as well as firm $i_{1}$, swap-1 manipulation by $j_{2}$ at $w$ is not successful. Since $w$ and $j_{2}$ were arbitrarily chosen, with the only requirement that the assigned type of $j_{2}$ is second-ranked in $w$, this
establishes that the AMIPA mechanism $\langle \sigma ^{A},\alpha ^{I}\rangle$ is $(2,1)$-non-manipulable and completes our argument for Statement A. 
\end{proof}

\noindent 
\textsc{Statement B}\hspace{0.2 cm}(\textbf{Induction Hypothesis}) \textit{Fix a $p'\in \{3,\cdots,n\}$. Suppose that the AMIPA mechanism $\langle \sigma^A,\alpha^I\rangle$ is $(p,k)$-non-manipulable for all $p<p'$. Then $\langle \sigma^A,\alpha^I\rangle$ is $(p',k)$-non-manipulable. }

\medskip 
\begin{proof}
Notice first that, in the definition of $(p,k)$%
-non-manipulability, the variable $k$ runs from $1$ to $p$. Hence, it is
important to remember that the $k$'s in the antecedent and the conclusion in
Statement B, i.e., in $(p,k)$ and $(p^{\prime },k)$, are different - they
vary over sets of different cardinality.

Fix a $p^{\prime }\in \{3,\cdots ,n\}$ and suppose that $\langle \sigma
^{A},\alpha ^{I}\rangle $ is $(p,k)$-non-manipulable for all $p<p^{\prime }$%
. We will argue that, for no type assignment $w\in \mathcal{W}$ and $k\leq
p^{\prime }-1$, there is some worker $j$ who is at position $p^{\prime }$ in 
$w$ and can perform a successful swap-$k$ manipulation.

We will complete the argument in two steps. In Step 1 we establish that,
for any $k<p^{\prime }-1$, there does not exist a type assignment $w\in 
\mathcal{W}$ and a worker $j$ at position $p^{\prime }$ such that $j$ can
perform a successful swap-$k$ manipulation. In Step 2, we establish the same
for $k=p^{\prime }-1$. Although the arguments in the two cases are very
similar, there is one crucial case that is present in Step 1 and not in Step
2 that we would like to indicate.\medskip \newline
\noindent \textsc{Step 1 }In this step we establish that the mechanism $%
\langle \sigma ^{A},\alpha ^{I}\rangle $ is $(p,k)$-non-manipulable where $%
k\in \{1,\cdots ,p^{\prime }-2\}$. Suppose to the contrary that there exists
a type assignment $w\in \mathcal{W}$ and a $k\leq p^{\prime }-2$ such that
worker $j$ with $p_{j}(w)=p^{\prime }$ can successfully swap-$k$ manipulate
at $w$ via $t\in T$. Without loss of generality, let $j=j_{p^{\prime }}$ and 
$w=(t_{1},t_{2},\cdots ,t_{n})$ be such that $t_{\ell }=w(j_{\ell })$ for
all $\ell \in \left\{ 1,\cdots ,n\right\} $ and%
\begin{equation*}
t_{1}>t_{2}>\cdots >t_{p^{\prime }-k}>\cdots >t_{p^{\prime }-1}>t_{p^{\prime
}}>\cdots >t_{n}.
\end{equation*}%
Suppose that, at the type assignment $w$, worker $j_{p^{\prime }}$ of type $%
w(j_{p^{\prime }})$ is contemplating misreporting via $t>t_{p^{\prime
}-k}=w(j_{p^{\prime }-k})$. We will show that such a misreporting will not
be successful.

As noted above, worker $j_{p^{\prime }}$ reports $r(j_{p^{\prime }})=t$
under the assumption every other worker is reporting truthfully. So, from
the perspective of $j_{p^{\prime }}$, the profile of reports is $%
r=(w_{-j_{p^{\prime }}},t)=(t_{-p^{\prime }},t)$. The matching $\sigma
^{A}(r)$ would then be%
\begin{equation*}
\begin{split}
(i_{p^{\prime }-k},j_{p^{\prime }})& \\
(i_{p^{\prime }-k^{\prime }+1},j_{p^{\prime }-k^{\prime }})& \quad \text{for
each }k^{\prime }\in \{1,2,\cdots ,k\} \\
(i_{\ell },j_{\ell })& \quad \text{for each }\ell \in \{1,2,\cdots
,p^{\prime }-k,p^{\prime }+1,\cdots ,n\}.
\end{split}%
\end{equation*}%
Given the matching, the misreport $r(j_{p^{\prime }})=t$ of worker $%
j_{p^{\prime }}$, and this worker's true type $w(j_{p^{\prime }})$, firm $%
i_{p^{\prime }-k}$ observes the report profile $r$, the true type $%
w(j_{p^{\prime }})$ of its matched partner $j_{p^{\prime }}$, and that $%
j_{p^{\prime }}$ has misreported via $t$. Firm $i_{p^{\prime }-k}$ does not
know the true type assignment, but knows that each worker has reported
assuming all other workers have reported truthfully. In particular, worker $%
j_{p^{\prime }}$ has reported assuming $t_{-p^{\prime }}$ as the true type
assignment for all other workers.

Consider the information set $\Pi _{i_{p^{\prime }-k}}(w,\sigma
^{A}(r),\alpha _{i_{p^{\prime }-k}}^{I}(r))$ of firm $i_{p^{\prime }-k}$ and
any type assignment $w^{\prime }$ in that set. Recall that $w^{\prime
}(j_{p^{\prime }})=w(j_{p^{\prime }})$ and as noted above, for all $\ell
\geq p^{\prime }+1$,%
\begin{equation*}
t_{p^{\prime }}=w^{\prime }(j_{p^{\prime }})>t_{\ell }\geq w^{\prime
}(j_{\ell }).
\end{equation*}%
Let $G=\{j_{1},j_{2},\cdots ,j_{p^{\prime }-1}\}=\{j_{\ell }\in J\mid
t_{\ell }>w^{\prime }(j_{p^{\prime }})=w(j_{p^{\prime }})\}$. In other
words, $G$ is the set of workers who have made reports $r(j_{\ell })$ that
are strictly higher than the \textit{true} type of worker $j_{p^{\prime }}$.
First we argue that for any $w^{\prime }\in \Pi _{i_{p^{\prime
}-k}}(w,\sigma ^{A}(r),\alpha _{i_{p^{\prime }-k}}^{I}(r))$ it can never be
the case that $t>w^{\prime }(j_{p^{\prime }})=t_{p^{\prime }}>w^{\prime }(j)$
for all $j\in G$. Suppose to the contrary there exists a $w^{\prime }$ with $%
t>w^{\prime }(j_{p^{\prime }})=t_{p^{\prime }}>w^{\prime }(j)$ for all $j\in
G$. At such a profile, worker $j_{p^{\prime }}$ has the highest assigned
type. Then it is not worthwhile for $j_{p^{\prime }}$ to misreport via $t$.
Since $j_{p^{\prime }}$ has indeed misreported, the argument above leads to
the conclusion that $w^{\prime }\notin \mathcal{W}(w(j_{p^{\prime
}}),r(j_{p^{\prime }})=t)$ should be the case. By nontrivial updating, $w^{\prime }\notin \Pi _{i_{p^{\prime }-k}}(w,\sigma
^{A}(r),\alpha _{i_{p^{\prime }-k}}^{I}(r))$ follows.

In view of the above argument, we can conclude that for any type assignment $%
w^{\prime }\in \Pi _{i_{p^{\prime }-k}}(w,\sigma ^{A}(r),\alpha
_{i_{p^{\prime }-k}}^{I}(r))$ there exists $G_{1}\subset G$ such that for
all $j_{\ell }\in G_{1}$, $w^{\prime }(j_{\ell })>w^{\prime }(j_{p^{\prime
}})=w(j_{p^{\prime }})$. Let $\bar{G}_{1}=G\setminus G_{1}$. That is, $\bar{G%
}_{1}$ is the set of workers in $G$ whose assigned types in the type
assignment $w^{\prime }$ are lower than the true type of $j_{p^{\prime }}$.
In other words, for all $j\in \bar{G}_{1}$, 
\begin{equation*}
w(j_{p^{\prime }})>w^{\prime }(j).
\end{equation*}%
Let us consider now the worker $j_{p^{\prime }-1}\in G$. Notice that in the
report profile $r$, the report $r(j_{p^{\prime }-1})=t_{p^{\prime }-1}$ by
worker $j_{p^{\prime }-1}$ is the lowest possible report that is higher than
the true type of $j_{p^{\prime }}$. There are two mutually exclusive and
exhaustive cases to consider here:

A. $j_{p^{\prime }-1}\notin G_{1}$, and

B. $j_{p^{\prime }-1}\in G_{1}$.

We will argue that Case A is not possible.\medskip \newline
\textsc{Case A: }$j_{p^{\prime }-1}\notin G_{1}$: Let $|\bar{G}%
_{1}|=g_{1}\geq 1$. Then $|G_{1}|=p^{\prime }-(g_{1}+1)$. Notice that, at
such a type assignment $w^{\prime }$, the position of $j_{p^{\prime }}$ in $%
w^{\prime }$ is $p_{j_{p^{\prime }}}(w^{\prime })=p^{\prime
}-(g_{1}+1)<p^{\prime }$. There are two subcases to consider:\medskip

\noindent {\textsc{Subcase A}}(i):\textsc{\ }For all $j\in G_{1}$,%
\begin{equation*}
w^{\prime }(j)>t>w^{\prime }(j_{p^{\prime }})=w(j_{p^{\prime }}).
\end{equation*}%
Observe that at such a type assignment $w^{\prime }$, it is not
worthwhile for $j_{p^{\prime }}$ to misreport via $t$. Since $j_{p^{\prime
}} $ has indeed misreported, it must be that $w^{\prime }\notin \mathcal{W}%
(w(j_{p^{\prime }}),r(j_{p^{\prime }})=t)$ and hence, $w^{\prime }\notin \Pi
_{i_{p^{\prime }-k}}(w,\sigma ^{A}(r),\alpha _{i_{p^{\prime }-k}}^{I}(r))$
follows by nontrivial updating.\textbf{\medskip 
\newline
}\noindent {\textsc{Subcase A}(ii):} There exists $G_{2}\subset G_{1}$ such
that for all $j\in G_{2}$,%
\begin{equation*}
t>w^{\prime }(j)>w^{\prime }(j_{p^{\prime }})=w(j_{p^{\prime }}).
\end{equation*}%
As noted before, $p_{j_{p^{\prime }}}(w^{\prime })=p^{\prime
}-(g_{1}+1)<p^{\prime }$. Let $p^{\ast }=p_{j_{p^{\prime }}}(w^{\prime })$
and $k^{\ast }=|G_{2}|$. Note that $k^{\ast }\leq p^{\ast }$. This implies
that the misreport of $j_{p^{\prime }}$ at such a type assignment $w^{\prime
}$ is a $(p^{\ast },k^{\ast })$-manipulation by worker $j_{p^{\prime }}$.
However, by the \textquotedblleft if\textquotedblright\ part of the
Induction Hypothesis, the mechanism $\langle \sigma ^{A},\alpha ^{I}\rangle $
is $(p^{\ast },k^{\ast })$-non-manipulable contradicting our previous
statement. Hence, $w^{\prime }\notin \Pi _{i_{p^{\prime }-k}}(w,\sigma
^{A}(r),\alpha _{i_{p^{\prime }-k}}^{I}(r))$. This only leaves us with Case
B.\medskip \newline
\noindent \textsc{Case B: }For all $w^{\prime }\in \Pi _{i_{p^{\prime
}-k}}(w,\sigma ^{A}(r),\alpha _{i_{p^{\prime }-k}}^{I}(r))$,%
\begin{equation*}
w^{\prime }(j_{p^{\prime }-1})>w^{\prime }(j_{p^{\prime }}).
\end{equation*}%
Recalling that $j_{p^{\prime }-1}$ is matched to $i_{p^{\prime }}$ and $%
s_{p^{\prime }-k}>s_{p^{\prime }}$, we conclude that such a misreport by $%
j_{p^{\prime }}$ will lead to $(i_{p^{\prime }-k},j_{p^{\prime }-1})$
forming a blocking pair. As $j_{p^{\prime }}$ can perform the above analysis
as well as firm $i_{p^{\prime }-k}$, a swap-$k$ manipulation by worker $%
j_{p^{\prime }}$ with $k<p^{\prime }-1$ is not successful.\medskip \newline
\noindent \textsc{Step 2 }In view of Step 1, we only need to establish that
the mechanism $\langle \sigma ^{A},\alpha ^{I}\rangle $ is $(p,k)$%
-non-manipulable where $k=p^{\prime }-1$. Suppose to the contrary that there
exists a type assignment $w\in \mathcal{W}$ such that worker $j$ with $%
p_{j}(w)=p^{\prime }$ has a successful swap-$k$ manipulation at $w$ via $%
t\in T$. Without loss of generality, let $j=j_{p^{\prime }}$ and $%
w=(t_{1},t_{2},\cdots ,t_{n})$ be such that $t_{\ell }=w(j_{\ell })$ for all 
$\ell \in \left\{ 1,\cdots ,n\right\} $ and%
\begin{equation*}
t_{1}>t_{2}>\cdots >t_{p^{\prime }-k}>\cdots >t_{p^{\prime }-1}>t_{p^{\prime
}}>\cdots >t_{n}.
\end{equation*}%
Suppose that at the type assignment $w$, worker $j_{p^{\prime }}$ with type $%
w(j_{p^{\prime }})$ is contemplating misreporting via $t>t_{p^{\prime
}-k}=w(j_{p^{\prime }-k})$. We will show that such a misreporting will not
be successful.

As noted above, worker $j_{p^{\prime }}$ reports $r(j_{p^{\prime }})=t$
under the assumption every other worker is reporting truthfully. So, from
the perspective of $j_{p^{\prime }}$, the profile of reports is $%
r=(w_{-j_{p^{\prime }}},t)=(t_{-p^{\prime }},t)$. The matching $\sigma
^{A}(r)$ would be%
\begin{equation*}
\begin{split}
(i_{1},j_{p^{\prime }})& \\
(i_{2},j_{1})& \\
\vdots & \\
(i_{p^{\prime }},j_{p^{\prime }-1})& \\
(i_{\ell },j_{\ell })& \quad \text{for each }\ell \in \{p^{\prime }+1,\cdots
,n\}.
\end{split}%
\end{equation*}%
Given the matching, the misreport $r(j_{p^{\prime }})=t$ of worker $%
j_{p^{\prime }}$, and that worker's true type $w(j_{p^{\prime }})$, firm $%
i_{1}$ observes the report profile $r$, the true type $w(j_{p^{\prime }})$
of its matched partner $j_{p^{\prime }}$, and that $j_{p^{\prime }}$ has
misreported via $t$. Firm $i_{1}$ does not know the true type assignment,
but knows that each worker has reported assuming all other workers have
reported truthfully. In particular, worker $j_{p^{\prime }}$ has reported
assuming $t_{-p^{\prime }}$ as the true type assignment for all other
workers.

Consider now the information set $\Pi _{i_{1}}(w,\sigma ^{A}(r),\alpha
_{i_{1}}^{I}(r))$ of firm $i_{1}$ and any type assignment $w^{\prime }$
belonging to that set. Recall that $w^{\prime }(j_{p^{\prime
}})=w(j_{p^{\prime }})$ and as noted above, for all $\ell \geq p^{\prime }+1$%
,%
\begin{equation*}
t_{p^{\prime }}=w^{\prime }(j_{p^{\prime }})>t_{\ell }\geq w^{\prime
}(j_{\ell }).
\end{equation*}%
Let $G=\{j_{1},j_{2},\cdots ,j_{p^{\prime }-1}\}=\{j_{\ell }\in J\mid
t_{\ell }>w^{\prime }(j_{p^{\prime }})=w(j_{p^{\prime }})\}$. In other
words, $G$ is the set of workers who have made reports $r(j_{\ell })$ that
are strictly higher than the \textit{true} type of worker $j_{p^{\prime }}$.
Analogous to Step 1, we can argue that for any $w^{\prime }\in \Pi
_{i_{1}}(w,\sigma ^{A}(r),\alpha _{i_{1}}^{I}(r))$ it can never be the case
that $t>w^{\prime }(j_{p^{\prime }})=t_{p^{\prime }}>w^{\prime }(j)$ for all 
$j\in G$. In view of this, we can conclude that for any type assignment $%
w^{\prime }\in \Pi _{i_{1}}(w,\sigma ^{A}(r),\alpha _{i_{1}}^{I}(r))$ there
exists a $G_{1}\subseteq G$ such that for all $j_{\ell }\in G_{1}$, $%
w^{\prime }(j_{\ell })>w^{\prime }(j_{p^{\prime }})=w(j_{p^{\prime }})$. Let 
$\bar{G}_{1}=G\setminus G_{1}$. That is, $\bar{G}_{1}$ is the set of workers
in $G$ whose assigned types in the type assignment $w^{\prime }$ are lower
than the true type of $j_{p^{\prime }}$. In other words, for all $j\in \bar{G%
}_{1}$, 
\begin{equation*}
w(j_{p^{\prime }})>w^{\prime }(j).
\end{equation*}%
As in Step 1, let us consider now the worker $j_{p^{\prime }-1}\in G$.
Notice that in the report profile $r$, the report $r(j_{p^{\prime
}-1})=t_{p^{\prime }-1}$ by worker $j_{p^{\prime }-1}$ is the lowest
possible report that is higher than the true type of $j_{p^{\prime }}$.
There are two mutually exclusive and exhaustive cases to consider here:

A. $j_{p^{\prime }-1}\notin G_{1}$, and

B. $j_{p^{\prime }-1}\in G_{1}$.

Our argument for this step will be complete if we show that Case A is not
possible.\medskip \newline
\noindent \textsc{Case A: }$j_{p^{\prime }-1}\notin G_{1}$: Let $|\bar{G}%
_{1}|=g_{1}\geq 1$. Then $|G_{1}|=p^{\prime }-(g_{1}+1)$. Notice that, at
such a type assignment $w^{\prime }$, the position of $j_{p^{\prime }}$ in $%
w^{\prime }$ is $p_{j_{p^{\prime }}}(w^{\prime })=p^{\prime
}-(g_{1}+1)<p^{\prime }$. Since worker $j_{p^{\prime }}$ has made the
highest report $t$\footnote{%
Notice that subcase A(i) in the previous Step 1 does not apply here.}, we
have for all $j\in G_{1}$,%
\begin{equation*}
t>w^{\prime }(j)>w^{\prime }(j_{p^{\prime }})=w(j_{p^{\prime }}).
\end{equation*}%
Let $p^{\ast }=p_{j_{p^{\prime }}}(w^{\prime })$, $k^{\ast }=|G_{1}|$, and
note that $k^{\ast }\leq p^{\ast }$. By $p_{j_{p^{\prime }}}(w^{\prime
})=p^{\prime }-(g_{1}+1)<p^{\prime }$, this implies that the misreport of $%
j_{p^{\prime }}$ at such a type assignment $w^{\prime }$ is a $(p^{\ast
},k^{\ast })$-manipulation by worker $j_{p^{\prime }}$. However, by the
\textquotedblleft if \textquotedblright\ part of the Induction Hypothesis,
the mechanism $\langle \sigma ^{A},\alpha ^{I}\rangle $ is $(p^{\ast
},k^{\ast })$-non-manipulable contradicting our previous statement. Hence, $%
w^{\prime }\notin \Pi _{i_{1}}(w,\sigma ^{A}(r),\alpha _{i_{1}}^{I}(r))$.
This only leaves us with Case B.\medskip \newline
\noindent \textsc{Case B: }For all $w^{\prime }\in \Pi _{i_{1}}(w,\sigma
^{A}(r),\alpha _{i_{1}}^{I}(r))$,%
\begin{equation*}
w^{\prime }(j_{p^{\prime }-1})>w^{\prime }(j_{p^{\prime }}).
\end{equation*}%
Recalling that $j_{p^{\prime }-1}$ is matched to $i_{p^{\prime }}$ and $%
s_{1}>s_{p^{\prime }}$, we conclude that such a misreport by $j_{p^{\prime
}} $ will lead to $(i_{p^{\prime }-k},j_{p^{\prime }-1})$ forming a blocking
pair. As $j_{p^{\prime }}$ can perform the above analysis as well as firm $%
i_{p^{\prime }-k}$, a swap-$k$ manipulation by worker $j_{p^{\prime }}$ with 
$k=p^{\prime }-1$ is not successful.\medskip

\noindent Since $w$ was arbitrarily chosen, this concludes our argument for
Statement B.
\end{proof}

Statement A and Statement B complete the proof of Theorem \ref{Theorem1}. 

\section{Proof of Theorem \ref{Theorem2}}\label{Proof-theorem2} 

\begin{proof}
Recall that in the environment we consider in Theorem \ref{Theorem2}, $\mathcal{R}=T^n$ holds. Let $\tau _{0}:\{1,\cdots ,n\}\rightarrow \{1,\cdots ,n\}$ be the \textit{%
identity} permutation, that is $\tau _{0}(k)=k$ for each $k\in \{1,\cdots
,n\}$. Given $T=\{t_{1},\cdots ,t_{L}\}$, we define the \textit{expanded}
type set $T(\tau _{0})$ by $T(\tau _{0})=\{t_{\ell ,k}\}_{{\ell \in
\{1,\cdots ,L\}},{k\in }\left\{ 1,\cdots ,n\right\} }$ where

\begin{itemize}
\item for each $\ell \in \{1,\cdots ,L\}$ and for all $j,k\in \{1,\cdots
,n\} $ with $j<k$, 
\begin{equation*}
t_{\ell ,j}>t_{\ell ,k};
\end{equation*}

\item for all $\ell ,m\in \{1,\cdots ,L\}$ with $\ell <m$ and for all $%
j,k\in \{1,\cdots ,n\}$, 
\begin{equation*}
t_{\ell ,j}>t_{m,k}.
\end{equation*}
\end{itemize}

Consider now the AMIPA mechanism $\langle \sigma _{\tau }^{A},\alpha
^{I}\rangle $ and pick a type assignment $w\in \mathcal{W}$. Given the
ordering $\tau $ fixed in the mechanism, we define the \textit{expanded}
type assignment $w^{\tau }:J\rightarrow T(\tau _{0})$ as follows: for each $%
k\in \{1,\cdots ,n\}$,%
\[
\lbrack w(j_{k})=t_{\ell }]\Leftrightarrow \lbrack w^{\tau }(j_{k})=t_{\ell
,\tau (k)}].
\]
We denote by $\mathcal{W}(\tau )$ the set of all such expanded type assignments.

Let $r\in T^{n}$ be a report profile where there are exactly two workers $%
j_{1}$ and $j_{2}$ who have made the same announcement $t_{\ell }$, i.e., $%
r(j_{1})=r(j_{2})=t_{\ell }$. Moreover, let $\tau $ be such that $\tau
(2)\rhd \tau (1)$. As we will make clear now, a report profile $r$ where
exactly two reports coincide is the only form of report coincidence that is
relevant in our model. Note that, at the time of making a report, any worker 
$j$ assumes that every other worker is reporting truthfully. Consequently,
from the perspective of worker $j$, the remainder of the report profile $%
r_{-j}$ will be a valid type assignment for the workers in $J\setminus \{j\}$, i.e., $r(j_{\ell })\neq r(r_k)$ for all $j_{\ell }\neq j_k\neq j$.  Hence, it suffices to consider report profiles where exactly two workers
make the same report. For such a report profile $r$, we define the \textit{augmented} report profile $r^{\tau }$ as follows: for each $j\in J$,
\[
\lbrack r(j_{k})=t_{\ell }]\Leftrightarrow \lbrack r^{\tau }(j_{k})=t_{\ell
,\tau (k)}].
\]
Note that there is a one-to-one correspondence between $r$ and $r^{\tau }$.
Moreover, $r^{\tau }$ is a valid type assignment, i.e., $r^{\tau }\in \mathcal{W}(\tau )$. Suppose now, at the type assignment $w\in \mathcal{W}$, worker $j_{2}$ contemplates misreporting via some $%
t_{\ell }=w(j_{1})=r(j_{1})>w(j_{2})$. As mentioned before, $w$ corresponds
to the unique expanded type assignment $w^{\tau }\in \mathcal{W}(\tau )$. We
now consider the \textit{auxiliary problem} where at the type assignment $w^{\tau }$%
, worker $j_{2}$ is misreporting via $t_{\ell ,\tau (2)}$. This leads to the
augmented report profile $r^{\tau }$ defined above. Furthermore, this
expanded report profile $r^{\tau }$ is a valid type assignment in $\mathcal{W%
}(\tau )$ that \textit{uniquely} corresponds to the report profile $r$. In the auxiliary problem, by
Theorem \ref{Theorem1}, worker $j_{2}$ cannot successfully
misreport at $w^{\tau }$ via $t_{\ell ,\tau (2)}$. Since the original problem and the auxiliary problem are \textit{isomorphic} to each other, worker $j_{2}$ cannot successfully manipulate the AMIPA mechanism $%
\langle \sigma _{\tau }^{A},\alpha ^{I}\rangle $ at $w$ via $t$, either. 
\end{proof}

\end{appendices}

\end{document}